\documentclass{style/vldb}
\usepackage{amssymb}
\usepackage{amsmath}
\usepackage{color}
\usepackage{times}
\usepackage{cite}
\usepackage{enumitem}
\usepackage{arydshln}
\usepackage{url}
\usepackage{algorithm}
\usepackage{algorithmic}

\usepackage{flushend}
\usepackage{subcaption}
\usepackage{balance}
\usepackage{enumitem}
\newtheorem{theorem}{Theorem}

\newtheorem{example}{Example}

\newcommand{\twodray}{{\sc 2draysweep }}
\newcommand{\twodonline}{{\sc 2donline }}
\newcommand{\hyperpolar}{{\sc hyperpolar }}
\newcommand{\arrangement}{{\sc satregions }}
\newcommand{\arrangementn}{{\sc satregions}}
\newcommand{\mdbaseline}{{\sc mdbaseline }} 
\newcommand{\anglepartitioning}{{\sc anglepartitioning }}  
\newcommand{\harrangement}{{\sc AT$_+$ }} 
\newcommand{\cellplane}{{\sc cellplane$_\times$ }} 
\newcommand{\cellarrangement}{{\sc markcell }} \newcommand{\harrangementc}{{\sc ATC$_+$ }} 
\newcommand{\cellcoloring}{{\sc cellcoloring }}
\newcommand{\mdonline}{{\sc mdonline }}
\newcommand{\twodraynsp}{{\sc 2draysweep}}
\newcommand{\twodonlinensp}{{\sc 2donline}}

\newcommand{\harrangementnsp}{{\sc AT$_+$}} 
\newcommand{\cellplanensp}{{\sc cellplane$_\times$}} 
\newcommand{\cellarrangementnsp}{{\sc markcell}}

\newcommand{\rankswitch}{ordering exchange }
\newcommand{\rankswitches}{ordering exchanges }
\newcommand{\rankswitchnospace}{ordering exchange}
\newcommand{\rankswitchesnospace}{ordering exchanges}
\newcommand{\stitle}[1]{\vspace{1ex}\noindent{\bf #1}}

\usepackage{amsmath}
\usepackage{graphicx}
\usepackage{comment}
\usepackage[colorinlistoftodos]{todonotes}
\usepackage[colorlinks=true, allcolors=blue]{hyperref}

\newcommand\julia[1]{\textcolor{blue}{(Julia) #1}}

\newcommand\techrep[1]{#1}
\newcommand\submit[1]{}

\title{Designing Fair Ranking Schemes}
\author{
	\alignauthor
	Abolfazl Asudeh$^\dag$, H. V. Jagadish$^\dag$, Julia Stoyanovich$^\ddag$, Gautam Das$^{\dag\dag}$ \\
	\affaddr {$^\dag$University of Michigan, $^\ddag$Drexel University, $^{\dag\dag}$University of Texas at Arlington}
	{\email{
		$^\dag\{$asudeh, jag$\}$@umich.edu, $^\ddag$stoyanovich@drexel.edu, $^{\dag\dag}$gdas@uta.edu
        }
	}
}

\begin{document}
\maketitle
\begin{abstract}
Items from a database are often ranked based on a combination of multiple criteria.  A user may have the flexibility to accept combinations that weigh these criteria differently, within limits.  On the other hand, this choice of weights can greatly affect the fairness of the produced ranking.  In this paper, we develop a system that helps users choose criterion weights that lead to greater fairness.

We consider ranking functions that compute the score of each item as a weighted sum of (numeric) attribute values, and then sort items on their score.  Each ranking function can be expressed as a vector of weights, or as a point in a multi-dimensional space. For a broad range of fairness criteria, we show how to efficiently identify regions in this space that satisfy these criteria.  Using this identification method, our system is able to tell users whether their proposed ranking function satisfies the desired fairness criteria and, if it does not, to suggest the smallest modification that does.
We develop user-controllable approximation that and indexing techniques that are applied during preprocessing, and support sub-second response times during the online phase.
Our extensive experiments on real datasets demonstrate that our methods are able to find solutions that satisfy fairness criteria effectively and efficiently.   
% \jag{Perhaps a sentence or so on the main experimental results??}

\end{abstract}

\section{Introduction}
\label{sec:intro}

Data-driven algorithmic decisions are commonplace today.
%in the public sector and in the commercial sphere.  
Because of the impact these decisions have on individuals and on population groups, issues of algorithmic bias and discrimination are coming to the forefront of societal and technological discourse~\cite{BarocasSelbst}.  In the seminal work of Friedman and Nissenbaum~\cite{DBLP:journals/tois/FriedmanN96} a biased computer system is one that (1) systematically and unfairly discriminates against some individuals or groups in favor of others, and (2) joins this discrimination with an unfair outcome. 

A prominent source of bias in data-driven systems is the data itself, a phenomenon colloquially known as ``racism in --- racism out''.  For example, it has been shown that machine learning models trained on biased data will produce biased results, further propelling historical discrimination~\cite{DBLP:journals/corr/EnsignFNSV17}. Naturally, the effects of biased data are not limited to machine learning scenarios, but also impact processes that are directly designed and validated by humans.  Perhaps the most immediate example of such a process is a score-based ranker.  In this paper we consider the task of {\em designing a fair score-based ranking scheme}. %, focusing on applications that rank individuals.  

Ranking of individuals is ubiquitous, and is used, for example, to establish credit worthiness, desirability for college admissions and employment, and attractiveness as dating partners. A prominent family of ranking schemes are score-based rankers, which compute the score of each individual from some database $\mathcal{D}$, sort the individuals in decreasing order of score, and finally return either the full ranked list, or its highest-scoring sub-set, the top-$k$.  Many score-based rankers compute the score of an individual as a linear combination of attribute values, with non-negative weights.  Designing a ranking scheme amounts to selecting a set of weights, one for each feature, and validating the outcome on the database $\mathcal{D}$.

Our goal is to assist the user in designing a ranking scheme that both reflects a user's a priori notion of quality and is fair, in the sense that it mitigates {\em preexisting bias with respect to   a protected feature} that is embodied in the data. In line with prior work~\cite{DBLP:conf/kdd/FeldmanFMSV15,StoyanovichYJ18,DBLP:conf/ssdbm/YangS17,DBLP:conf/cikm/ZehlikeB0HMB17,DBLP:journals/datamine/Zliobaite17}, a protected feature denotes membership of an individual in a legally-protected category, such as persons with disabilities, or under-represented groups by gender or ethnicity.  Interpreting the definition of Friedman and Nissenbaum~\cite{DBLP:journals/tois/FriedmanN96} for rankings, a biased outcome occurs when a ranking decision is based fully or partially on a protected feature.   Discrimination occurs when this outcome is systematic and unfavorable, for example, when minority ethnicity or female gender systematically lead to placing individuals at lower ranks.  To make our discussion concrete, we consider an example.

\begin{example}\label{ex:uni} 
A college admissions officer is evaluating a pool of applicants, each with several potentially relevant attributes.  For simplicity, let us focus on two of these attributes --- high school GPA 
%(grade point average) 
and SAT score %(US Scholastic Assessment Test)
, and use these in a score-based ranking scheme.

As the first step, to make the two score components comparable, GPA and SAT scores may be normalized and standardized.  We will denote the resulting values $g$ for GPA and $s$ for SAT.  
The admissions officer may believe a priori that $g$ and $s$ should have an approximately equal weight, computing the score of an applicant $t \in \mathcal{D}$ as $f(t) = 0.5 \times s + 0.5 \times g$, ranking the applicants, and returning the top 500 individuals.  

Upon inspection, it may be determined that an insufficient number of women is returned among the top-$k$: at least 200 women were expected to be among the top-$500$, and only 150 were returned, violating a fairness constraint.  This violation may be due to a gender disparity in the data: in 2014, women scored about 25 points lower on average than men in the SAT test~\cite{SAT2014}.  

The system will then assist the user in identifying a new scoring functions $f'(t) = 0.45 \times s + 0.55 \times g$, which meets the fairness constraint and is close to the original function $f$ in terms of attribute weights, thereby reflecting the user's a priori notion of quality.  

\end{example} 

In machine learning, the common setup is to have training data for which we know the outcome (label), and then the problem is to have the system learn weights (or other model parameters) that result in an algorithm that maximizes the predicted outcome (or correctness of label).  While this problem setup is appropriate for many tasks, it requires unreasonably simplistic assumptions in many others.  For example, what is the outcome an admissions officer seeks to maximize in admitting students?  Some outcomes are relatively easy to measure, such as GPA after admission and enrollment.  But what the university presumably really cares about is long-term success: the admissions officer is looking for students who will go on to become rich or famous or successful in some other dimension that matters.  This outcome is fuzzy, multi-dimensional, and hard to measure.  It is also long-term --- the algorithm cannot really be tuned for today based on data regarding students admitted 30 years ago.  For these reasons, many practical systems have simple models with weights set by human experts, usually in a subjective manner.

Precisely because these weights are often subjectively chosen, we have an even greater fear of discrimination than just algorithmic bias.  In fact, there is a long history of people using justifiable models to be able to discriminate.  For example, legacy was added to the variables considered at admission, and given a high weight, to keep down the number of Jewish students, since ``too many'' of them would have been admitted considering academic achievements alone~\cite{jacobs13,karabel05}.

In this paper, we consider this sort of reverse problem: the selection of model weights after we already know (the distribution of attribute values in) the dataset to be scored and ranked.  The goal is to select weights such that desired fairness and diversity criteria are satisfied.   To be certain of meeting these criteria, the weights have to be selected after we have the dataset in hand.   If we know that the distribution of values in the dataset will not change too much over some window, we can go through a design process to choose model weights once using a representative sample of the data, and then just reuse the same model and weights for each dataset that follows.  We may still wish to verify that we continue to meet the required criteria, and adjust our ranking function if needed.  In short, the choice of ranking function is not a one-time thing.  Rather, in practice, ranking functions are frequently tuned, typically with small changes.

As such, we repeatedly have a human model designer trying to tune model weights.  It may be acceptable for this tuning process to take some time. However, we know that humans are able to produce superior results when they get quick feedback in a design or analysis loop.  Indeed, it is precisely this need that is a central motivation for OLAP, rather than having only long-running analytics queries.   Ideally, a human designer of a ranking function would want the system to support her work through interactive response times.  Our goal is to meet this need, to the extent possible.

In the remainder of this paper, we will present a {\em query answering system} that assists the user in designing fair score-based rankers.  As the first step, the system pre-processes a dataset of candidates off-line, and is then able to handle user requests in real time.  The user specifies a {\em query} in the form of a scoring function $f$, which associates non-negative weights with item attributes and is used to compute items scores, and to sort the items on their scores.  We assume the existence of a {\em fairness oracle} that, given an ordered list of items, returns $true$ if the list meets fairness criteria and so is {\em satisfactory}, and returns $false$ otherwise.  If the list of items was found to be unsatisfactory, we will suggest to the user an alternative scoring function $f'$ that is both satisfactory and close to the query $f$.  The user may accept the suggested function $f'$, or she may decide to manually adjust the query and invoke our system once again. 

Numerous fairness definitions have been considered in the literature~\cite{DBLP:conf/innovations/DworkHPRZ12,DBLP:journals/datamine/Zliobaite17}. A useful dichotomy is between {\em individual fairness}, and {\em group fairness}, also known as statistical parity. The former requires  that similar individuals be treated similarly, while the latter requires that demographics of those receiving a particular outcome are identical or similar to the demographics of the population as a whole~\cite{DBLP:conf/innovations/DworkHPRZ12}.  These two requirements represent intrinsically different world views, and accommodating both may require trade-offs~\cite{DBLP:journals/corr/FriedlerSV16}. Our focus is on group fairness, which is based on the relationship between (1) membership of individuals in demographic groups and (2) their ranked outcome.

While fairness in algorithmic systems is an active area of research~\cite{DBLP:journals/datamine/Zliobaite17}, our work is among a small handful of studies that focus on fairness in ranking~\cite{DBLP:journals/corr/CelisSV17,DBLP:conf/ssdbm/YangS17,DBLP:conf/cikm/ZehlikeB0HMB17}.  While others considered mitigating bias in the output of a ranker~\cite{DBLP:journals/corr/CelisSV17,DBLP:conf/cikm/ZehlikeB0HMB17}, or incorporating fairness constraints into ranked models~\cite{DBLP:conf/ssdbm/YangS17}, {\em our work is the first to support the user in designing fair ranking schemes}.  

Our methods are general, and can accommodate a large class of group fairness constraints --- including those based on asserting a minimum or a maximum number of individuals at the top-$k$ that belong to a particular demographic group, as in Example~\ref{ex:uni} and in the work of Celis et al.~\cite{DBLP:journals/corr/CelisSV17}, but going far beyond this class.  In fact, our techniques treat the evaluation of fairness constraints as a black box (embodied by the fairness oracle), and support any constraint that can be evaluated over a ranked list of items.  We are not limited to binary protected group membership (such as gender in Example~\ref{ex:uni}, and in the work of Zehlike et al.~\cite{DBLP:conf/cikm/ZehlikeB0HMB17}), and can accommodate fairness constraints on multiple non-overlapping population groups (such as ethnicity).  Further, we support constraints that are stated over more than one sensitive attribute; for example, we can enforce constraints on gender, ethnicity and age group simultaneously.

\vspace{2mm} 
\noindent{\bf Summary of contributions:} In this paper, we assist the user in designing fair score-based ranking schemes.  Towards this goal, we characterize the space of linear scoring functions (with their corresponding weight vectors), and characterize portions of this space based on the ordering of items induced by these functions. We develop algorithms to determine boundaries that partition the space into regions where any desired fairness constraint is satisfied, called {\em satisfactory regions}, and regions where the constraint is not satisfied. Given a user's query, in the form of a scoring function $f$, we develop techniques to find the nearest scoring function $f'$ that satisfies the constraint (or to state that the constraint is not satisfiable).  Our 
%algorithmic 
contributions are as follows: %\abol{We may need to provide the necessary technical highlights. For examples the fact that we consider a fixed number of dimensions is important; because consequently nothing is NP-complete, all the proposed solutions polynomial. Or that we propose the approximate approx bcz we want to work as an online alg. not bcz it is NP-complete.\\} \julia{I wouldn't talk about that in the intro, but let's discuss.}
%\abol{More importantly, I believe the connection b/w the motivation exmaples/highlights, and the system architecture is lost. Here we do not even mentioned that we preprocess the data offline and our objective is to be fast in online answering of the user queries. It is not even clear {\em why we design a query answering system}. As by reading the examples and motivations (and even the title) it seems like it is a single shot thing. If it is for a single shot, what is the purpose of preprocessing the data.} \julia{I rephrased the paragraph right after Example 1, take a look.}

\begin{itemize}[itemsep=0pt,leftmargin=10pt]
\item We propose a query answering system that helps users choose ranking functions that meet fairness requirements. (\S~\ref{sec:intro})

\item Carefully defining the terms, problem statement, and assumptions (\S~\ref{sec:pre}), we pursue offline indexing of the data that helps in online answering of users' queries.

\item We introduce the notion of \rankswitch and use a transformation of the items to dual space to identify satisfactory regions, in which fairness constraints are met.  (\S~\ref{sec:2d})

\item We propose the ray sweeping algorithm \twodray for indexing the satisfactory regions in two-dimensional space and \twodonlinensp, an exact logarithmic binary search-based algorithm that takes as input a user's query $f$ and proposes an alternative $f'$ in real time. (\S~\ref{sec:2d}) %for the online answering of the user queries.

\item For studying the linear ranking functions in fixed-size multi dimensional spaces, we introduce an angle coordinate system. We propose \hyperpolar to transform the \rankswitches in the angle coordinate system.  (\S~\ref{sec:md})

\item We use ``the arrangement of hyperplanes''~\cite{edelsbrunner, orlik2013arrangements} and propose the polynomial time exact algorithms \arrangement and \mdbaseline for identifying satisfactory regions and proposing fair scoring functions to a user in a space with arbitrary number of dimensions. (\S~\ref{sec:md})

\item We propose the arrangement tree data structure for optimizing the running time of \arrangementn. (\S~\ref{sec:md})

\item We propose a user-controllable grid partitioning of the angle coordinate system that guarantees a maximum angle distance between every pair of points in a cell.

\item We use the grid partitioning of the angle coordinate system and propose an approximate algorithm (for multi-dimensional space) that guarantees a controllable distance from the optimal solution, while enabling efficient processing of users' queries. This approximate algorithm provides opportunities for speeding up the indexing algorithms by limiting the arrangements to each cell and applying an early stopping strategy. We propose algorithms \cellplanensp, \cellarrangementnsp, and \cellcoloring for indexing, and an efficient online algorithm \mdonline for answering users' queries. (\S~\ref{sec:speedup})
 
%(iv) We develop the BFS algorithm \cellcoloring to assign a satisfactory function to every cell that does not contain a satisfactory function.

%\item For large-scale datasets, we propose a sampling-based strategy that provides a ``close to satisfactory'' output with high probability.
%\item In addition to the theoretical analyses, we conduct extensive experiments on real datasets that confirm the efficiency and effectiveness of our proposal in practice.
\end{itemize}

In addition to the theoretical analyses, we conduct extensive experiments on real datasets that confirm the efficiency and effectiveness of our techniques, as described in \S~\ref{sec:exp}.
Related work and conclusions are discussed  in \S~\ref{sec:related} and \S~\ref{sec:conclusion}, respectively.

%We begin by setting up the problem in \S~\ref{sec:pre}. 
%We then consider the simple 2D case in \S~\ref{sec:2d}, followed by a study of the  general MD case in \S~\ref{sec:md}. In \S~\ref{sec:speedup}, we propose several techniques for making the MD solutions practical. 

\section{Preliminaries} \label{sec:pre}

\noindent{\bf Data model:}
We are given a dataset $\mathcal{D}$ of $n$ items, each with $d$ scalar scoring attributes\footnote{\small Additional non-scalar attributes are considered in the fairness model.}. 
We represent an item $t$ as a $d$-long vector of scoring attributes, $\{t[1], t[2], \ldots, t[d]\}$.
Without loss of generality, we assume that each scoring attribute is a non-negative number and that larger values are preferred.  This assumption is straightforward to relax with some additional notation and bookkeeping. 

%In addition to $d$ scoring attributes, items in $\mathcal{D}$ have one or more type attributes (e.g., gender, ethnicity, or age group).  Each type attribute is discrete and categorical, and it partitions the items in $\mathcal{D}$. We will state fairness constraints based on the order in which items of different types appear in a ranking.

\vspace{2mm}
\noindent{\bf Ranking model:}
Our focus in this paper is on the class of linear ranking functions that use a weight vector $\vec{w} ~=~\{w_1, w_2,\ldots,w_d\}$ to compute a goodness score $f_{\vec{w}}(t)$\footnote{\small To simplify notation, we use $f(t)$ to refer to $f_{\vec{w}}(t)$.} of item $t$ as $\Sigma_{j=1}^d w_j t[j]$.
Without loss of generality, we assume each weight $w_j\in\vec{w} \geq 0$.
The scores of items are used for ranking them. 
We assume that an item with a higher score outranks an item with a lower score.

Our ranking model has an intuitive geometric interpretation: items are represented by points in $\mathbb{R}^d$, and a linear scoring function $f$ is represented by a ray starting from the origin and passing through the point $\vec{w} = \{w_1, w_2,...,w_d\}$.   The score-based ordering of the points induced by $f$ corresponds to the ordering of their projections onto the ray for $\vec{w}$. %in $\mathbb{R}^d$ on the corresponding ray for $\vec{w}$.
Figure~\ref{fig:2d} shows the items of an example dataset with $d=2$ as points in $\mathbb{R}^2$.
The function $f=x+y$ is represented in Figure~\ref{fig:2d1} as a ray stating from the origin and passing through the point $\{ 1,1 \}$. Projections of the points onto the ray specify their ordering based on $f$.

Note that the rays corresponding to functions $f$ and $f'$ are the same if the weight vector of $f'$ is a linear scaling of the weight vector of $f$. This is because a weight vector $\vec{w} = \{w_1, w_2,\ldots, w_d\}$ induces the same ordering on the items as does its linear scaling $\vec{w'} = \{c.w_1, c.w_2,\ldots,c.w_d\}$, for any $c>0$. %, equivalent with $\vec{w}$.
%That is because these functions always generate the same ordering of items and are considered as equivalent.
Hence, the {\em distance between two functions} $f$ and $f'$ is considered as the angular distance between their corresponding rays in $\mathbb{R}^d$.
For example, the distance between $f=x+y$ and $f'=100x+100y$ is 0, while the distance between $f=x+y$ and $f''=x$ is $\frac{\pi}{4}$, the angular distance between the ray corresponding to $f$ in Figure~\ref{fig:2d1} and the $x$-axis.
\submit{Further details for computing the angular distance between two functions are provided in the technical report~\cite{techreport}.}
\techrep{Further details for computing the angular distance between two functions are provided in Appendix~\ref{appendix:angle}.}
For every item $t\in\mathcal{D}$, {\em contour} of $t$ on $f$ is the value combinations in $\mathbb{R}^d$ with the same score as $f(t)$~\cite{asudeh2016query, asudeh2017rrms}. For linear functions, the contour of an item $t$ is the hyperplane $h$ that is perpendicular to the ray of $f$ and passes through $t$. %\julia{What's a set of combinations in $\mathbb{R}^d$?}

\begin{figure}[t!]
    \centering
    \begin{subfigure}[b]{0.24\textwidth}
        \centering
        \includegraphics[width = 0.95\textwidth]{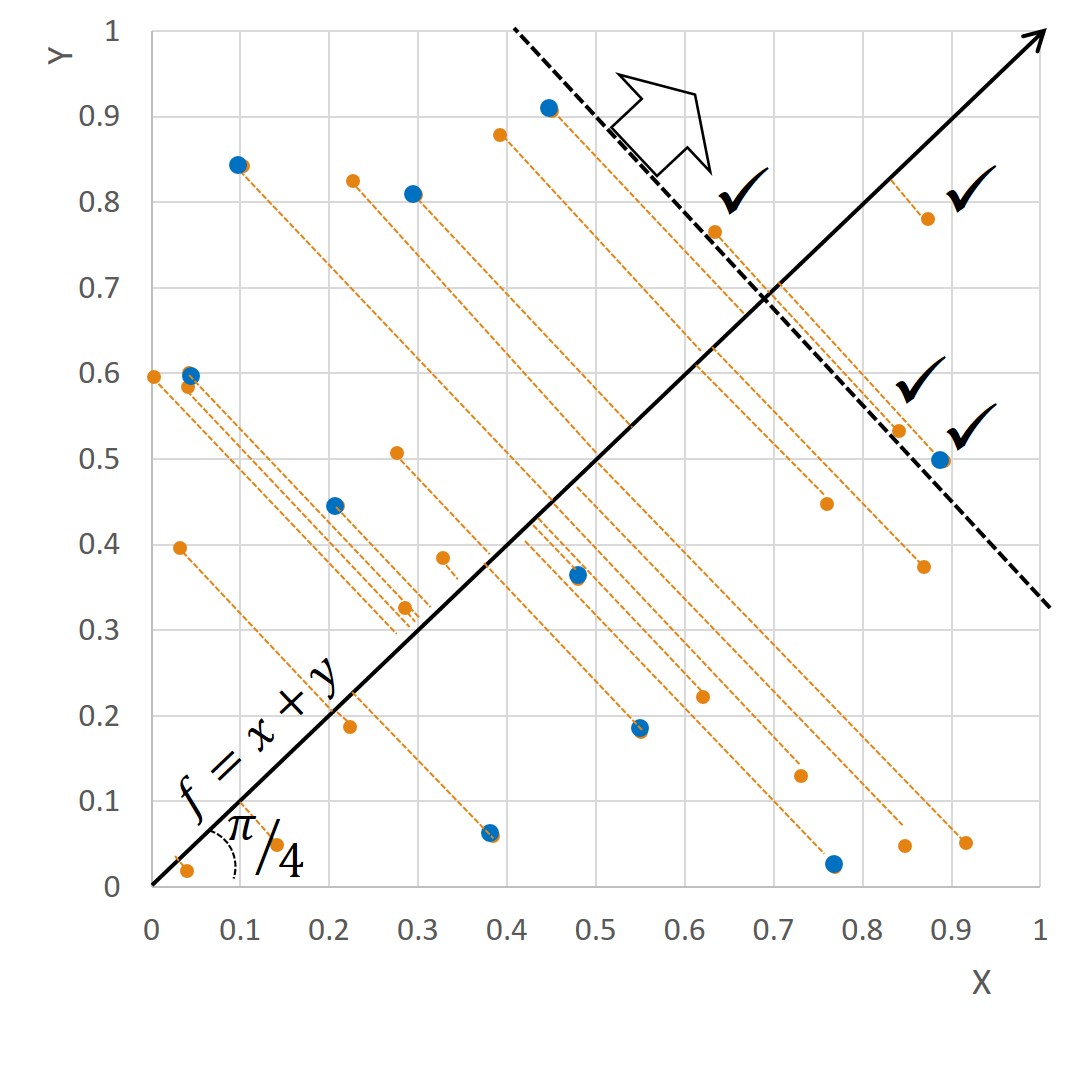}
        \caption{function $f=x+y$}\label{fig:2d1}
    \end{subfigure}%
    ~ 
    \begin{subfigure}[b]{0.24\textwidth}
        \centering
        \includegraphics[width = \textwidth]{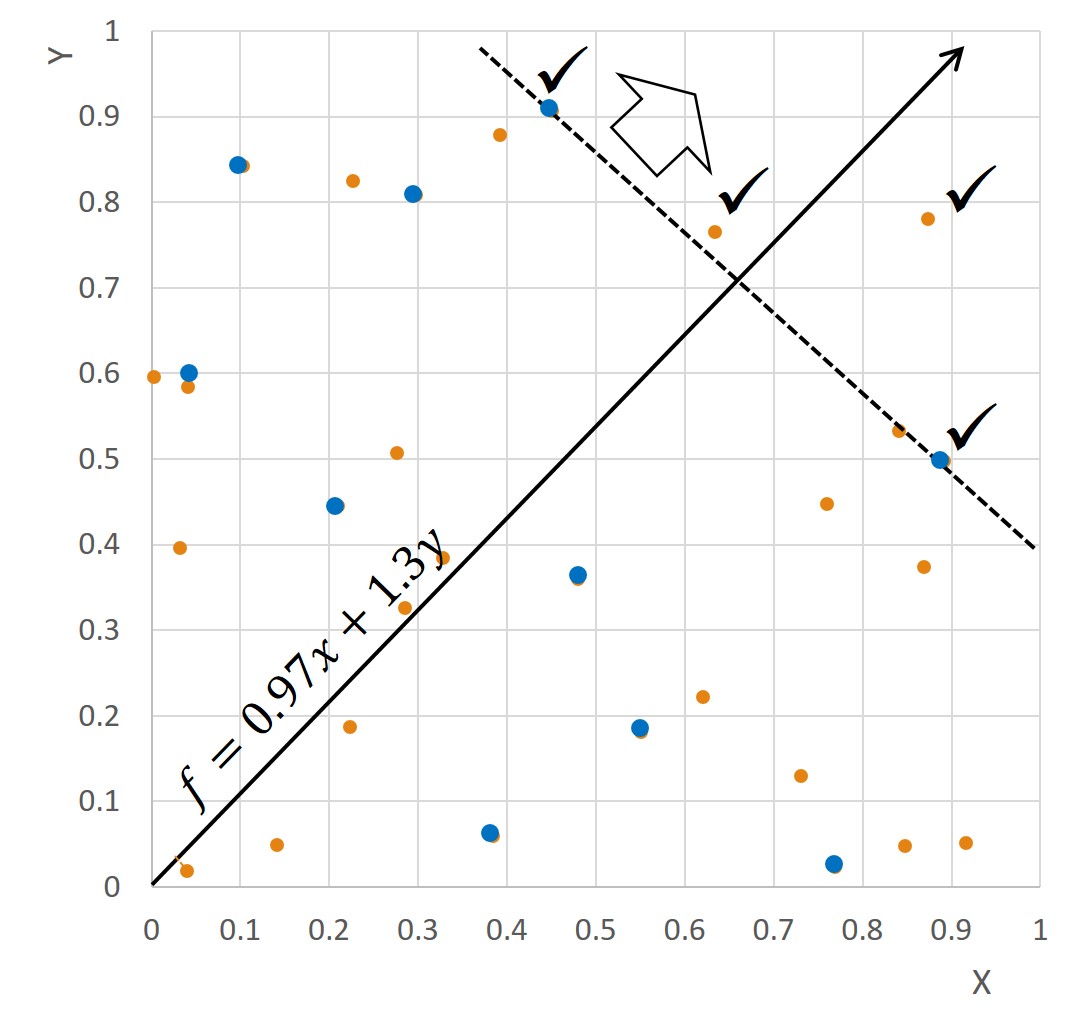}
        \caption{function $f=0.97x+1.3y$}\label{fig:2d2}
    \end{subfigure}
    \vspace{-4mm}\caption{2D, the effect of two similar functions on output fairness.}\label{fig:2d}
\end{figure}

%We will need to characterize similarity of weight vectors.  For this purpose, we compute the distance between two vectors, $\vec{v}$ and $\vec{w}$ as $\Sigma_{j=1}^d \{(w_j - v_j)^2/w_j v_j\}$.   We require normalization with the denominator because different dimensions may be scaled differently, and hence have very different weights.  Note that all weights are required to be non-zero.

%The ordering of items based on $\vec{w}$ is equal to the ordering of projection of their corresponding points in $\mathbb{R}^d$ on the corresponding ray for $\vec{w}$. Figure~\ref{fig:2d} shows the items of an example dataset where $d=2$, as the points in $\mathbb{R}^2$. The ray corresponding to the function $f=x+y$ is specified in Figure~\ref{fig:2d1} as a ray stating from the origin and passing through the point $\{1,1\}$. In addition, Figure~\ref{fig:2d1} shows projections of the points on the ray of $f$, which specify their ordering based on $f$.

\vspace{2mm}
\noindent{\bf Fairness model:} We adopt a general ranked fairness model, in which a fairness oracle $\mathcal{O}$ takes as input an ordered list of items from $\mathcal{D}$, and determines whether the  list meets fairness constraints: $\mathcal{O}: \mbox{ordered}(\mathcal{D}) \rightarrow \{\top, \bot\}$. A scoring function $f$ that gives rise to a fair ordering over  $\mathcal{D}$ is said to be {\em satisfactory}.

In addition to scoring attributes, discussed in the data model, items are associated with one or several type attributes. A type corresponds to a {\em protected feature} such as gender or race.  We discussed bias with respect to a protected feature in the introduction.
In the example in Figure~\ref{fig:2d}, there is a single binary type attribute, denoted by blue and orange colors.  Suppose that the fairness oracle returns true if the top-$4$ items contain an equal number of items of each type. Function $f=x+y$ in Figure~\ref{fig:2d1} is not satisfactory as it has 3 orange points and one blue point in its top-4, while $f'=0.97x + 1.3y$ in Figure~\ref{fig:2d2} contains two points each type in its top-$4$ and is satisfactory.

While our fairness model is general, in our experimental evaluation we focus on fairness constraints that were considered in recent literature~\cite{DBLP:journals/corr/CelisSV17,StoyanovichYJ18,DBLP:conf/cikm/ZehlikeB0HMB17}:  We work with proportionality constraints that bound the number of items belonging to a particular demographic group (as represented by an assignment of a value to a categorical type attribute) at the top-$k$, for some given value of $k$. 

%Zehlike et al.~\cite{DBLP:conf/cikm/ZehlikeB0HMB17} accommodate only binary membership in a protected group (such as gender) and make a statement about a ranking being fair based on the expected proportion of members of a protected group in a prefix of a top-$k$.  Celis et al.~\cite{DBLP:journals/corr/CelisSV17} quantify fairness in a ranked list as an upper bound on the number of items at the top-$k$ that belong to multiple, possibly overlapping, types.  For example, at most 5 males and at most 6 Asians can be at the top-$k$, and the same individual can be both male and an Asian.  We will use variants of these constraints in our experiments.

%Suppose that we select the top-$k$ items in the rank order, e.g. students for admission to a degree program at a university. Common fairness and diversity requirements can be expressed as constraints on the numbers of items selected in each class.  For example, we could check if there is proportionate representation of each class of items in the selected set.

%In the example provided in Figure~\ref{fig:2d}, in addition to the two scoring attributes, the items have a non-scoring binary attribute specified by blue and orange.

%\subsection{The origin-staring rays and the angle between them}\label{subsec:computeAngledistance}
%\subsection{Computing the The angle  between them}\label{subsec:computeAngledistance}

\subsection{Problem statement}\label{subsec:problemstate}
A given query $f$, with a corresponding weight vector, may not satisfy the required fairness constraints.  Our problem is to propose a scoring function $f'$ with a similar weight vector as $f$ %(with the smallest angle to the initial choice) 
that does satisfy the constraints, if one exists.  

Of course, the user may not accept our proposal.  Instead, she may try a different weight vector of her liking, which we can again examine and either approve or propose an alternative.  The final choice of an acceptable scoring function is up to the user. %Ultimately, it is the user's decision what to do.
The formal statement of our problem is as follows:

\vspace{0.02in}
\medskip\noindent
\framebox[\columnwidth]{\parbox{0.9\columnwidth}{ \textbf{\textsc{Closest Satisfactory Function:}}
\\ \textit{Given
a dataset $\mathcal{D}$ with $n$ items over $d$ scalar scoring attributes, %and $d'$ protected attributes,
%while each item $t\in\mathcal{D}$ belongs to the class $type(t)$,
a fairness oracle 
$\mathcal{O}: \mbox{ordered}(\mathcal{D}) \rightarrow \{\top, \bot\}$,
%$\mathcal{O}: \mbox{ordered}(\mathcal{D}) \rightarrow \{$true, false $\}$,
and a linear scoring function $f$ with the weight vector $\vec{w} = \{ w_1,w_2,\cdots ,w_d\}$, find the function $f'$ with the weight vector $\vec{w'}$ such that $\mathcal{O}($OrderBy$_{f'}(\mathcal{D}))=\top$  and the angular distance between $\vec{w}$ and $\vec{w'}$ is minimized.
}}}\\

%\julia{We have multiple types now, OK to remove type from the problem statement?}

\vspace{2mm}
\noindent{\bf High-level idea:}
From the system's viewpoint, the challenge is to propose similar weight vectors that satisfy the fairness constraints, in interactive time.  To accomplish this, our solution will operate with an offline phase and then an online phase.  In the offline phase, we will process the dataset, and develop data structures that will be useful in the online phase.  In the online phase, we will exploit these data structures to quickly find similar satisfactory weight vectors.  %We discuss these phases in turn in the following sections.
In the next section, we consider the easier to visualize 2D case, in which the dataset contains $2$ scalar scoring attributes.
%We first discuss how to identify and index the ranges of satisfactory weight vectors. To do so, we introduce the \rankswitch notion and use a dual transformation of items to design a ray-sweeping based algorithm \twodray.
The terms and techniques discussed in \S~\ref{sec:2d} will help us in \S~\ref{sec:md} for developing algorithms for the general multi-dimensional case where the number of scalar scoring attributes is $d>2$.
%Next, we provide a review of some geometric terms that will help in devising the algorithms. Then, in \S~\ref{sec:2d} and~\ref{sec:md} we study the problem for two dimensional and multi-dimensional cases where the number of scalar attributes are $d=2$ and $d>2$, respectively.

\section{The Two-Dimensional Case}\label{sec:2d}

In this section we consider a simplified version of the problem in which only two scalar attributes ($x$ and $y$) participate in the ranking.
The problem in 2D is easier to understand, visualize, and explain, and allows us to create the foundation for the general problem, which we will address in subsequent sections.  We begin by introducing the central notion of \rankswitch that partitions the space of linear functions into disjoint regions.  Then, we use this concept to develop two algorithms: an offline algorithm to identify and index the satisfactory regions, and an online algorithms that can be used repeatedly, as the domain expert interactively tunes weights, to obtain a desired ranking function.  

%In this section, as well as \S~\ref{sec:md}, our first objective is to identify and index the ranges of satisfactory functions, i.e., $\{f | \mathcal{O}($OrderBy$_{f}(\mathcal{D}))=$ true$\}$, in the offline preprocessing phase. Then, we aim to design efficient algorithms that use the indexed data and answer users' queries in an online manner.
%The critical requirement is that while preprocessing can be less efficient and take more time, the online query answering should be fast.
%In the following, we first introduce the notion of 
%. We use a dual transformation of the items for identifying the \rankswitches and propose a ray sweeping algorithm for identifying/indexing the satisfactory regions as a one-dimensional sorted list.
%Then, in \S~\ref{subsec:2donline}, we propose a logarithmic-time algorithm for the online answering of the user queries.

\subsection{Ordering exchange}\label{subsec:rankswitch}
Each item in a 2-dimensional dataset can be represented as a point in $\mathbb{R}^2$, and each ranking function $f$ can be represented as a ray starting from the origin. The ordering of the items is the ordering of their projections on the ray of $f$. For instance, Figure~\ref{fig:2d} specifies the projection of the points %in the figure 
on the ray of $f=x+y$.
One can see that the set of rays between the $x$ and $y$ axes represents the set of possible ranking functions in 2D.
Even though an infinite number of rays exists between $x$ and $y$, the number of possible orderings of $n$ items is limited to $n!$, the number of their permutations.  
Our central insight is that we do not need to consider every possible ranking function: we only need to consider at most as many as there are orderings of the items, as we discuss next. 

Consider two points $t_1\langle 1,2\rangle$ and $t_2\langle 2,1\rangle$, shown in Figure~\ref{fig:2dexplain}. The projections of $t_1$ and $t_2$ on the $x$-axis are the points $x=1$ and $x=2$, respectively. Hence, the ordering based on $f=x$ is $t_2\succ t_1$, which denotes that $t_2$ is preferred to $t_1$ by $f$. Moving away from the $x$-axis towards the $y$-axis, the distance between the projections of $t_1$ and $t_2$ on the ray decreases, and becomes zero at $f=x+y$. Then, moving from $f=x+y$ to the $y$-axis, the ordering between these two points changes to $t_1\succ t_2$. As we continue moving towards the $y$-axis, the distance between the projections of $t_1$ and $t_2$ increases, and their order remains  $t_1\succ t_2$. % and those do not switch ordering again.
Using this observation, we can partition the set of scoring functions based  on their angle with the $x$-axis into
%It thus partitions the set of ranking functions, based on their angle with the $x$ axis, to
$F_1=[0, \pi/4]$ and $F_2 = [\pi/4, \pi/2]$, such that for every $f\in F_1$ the ordering is $t_2 \succeq t_1$ and for every $f'\in F_2$ the ordering is $t_1 \succeq t_2$.
%\julia{The statement before was: ``$(t_2, t_1)$ and for every $f'\in F_2$ the ordering is $(t_1, t_2)$''.  This is (a) contrary to the order described above and (b) is imprecise, because you are not explicitly explaining what happens when neither of the points is preferred to the other. I suggest to use notation $t_1 \succ_{f} t_2$ when $t_1$ is ranked higher than $t_2$ by $f$, and $t_1 \succeq_{f} t_2$ when $t_1$ is ranked no lower than $t_2$. For example: $t_1 \succ_{x} t_2$, $t_2 \succ_{x+y} t_1$. Or even omitting the subscript: $t_1 \succ t_2$ is clearer than $(t_1, t_2)$.}
%\julia{Also, notation $F_1=[0, \pi/4]$ an then $f \in F_1$ is not clean.}\abol{I used capital F because it refers to a set of functions (every point in the range is a function).}

\begin{figure}[t!]
    \centering
    \includegraphics[width = 0.4\textwidth]{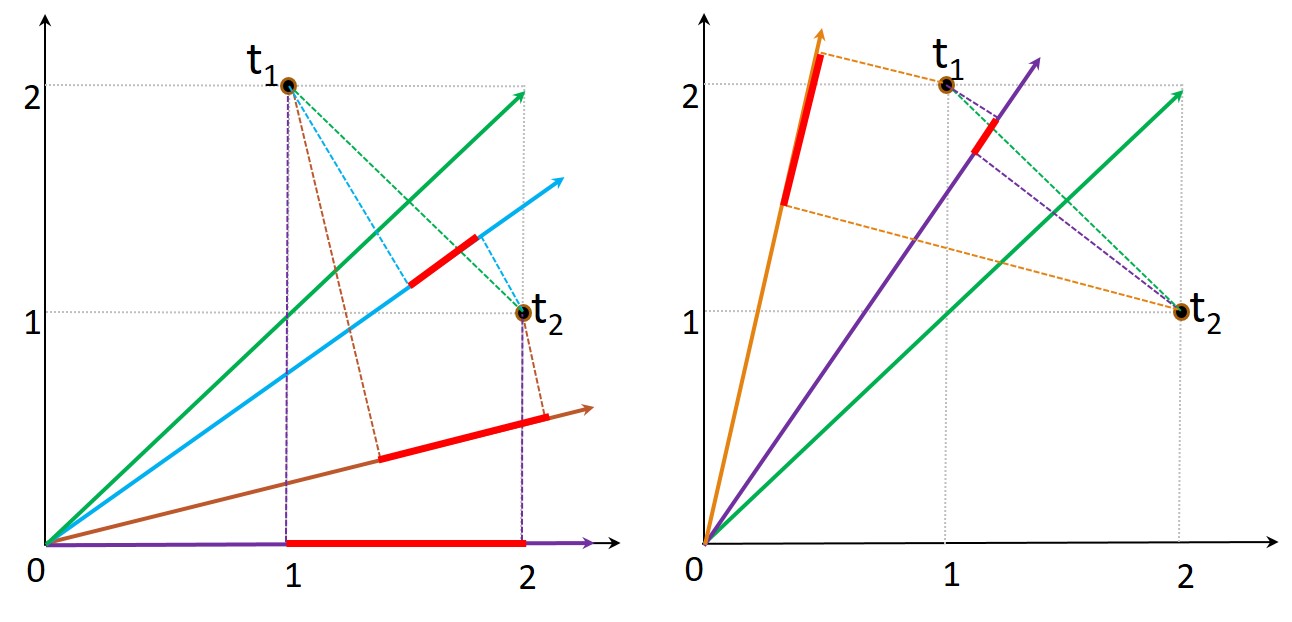}
    \vspace{-3mm}\caption{The \rankswitch between a pair of points. }\label{fig:2dexplain}
\end{figure}

Given the importance of the place where items $t_1$ and $t_2$ switch ordering, we define the {\em \rankswitch} as the ranking functions according to which $t_1$ and $t_2$ are equally good. In 2D, the \rankswitch of a pair of points is at most a single function.

For any specified ordering of items, the fairness constraint either is satisfied or it is not.  If this ordering is changed, the satisfaction of the fairness constraint may change as well.  Therefore, in the space of possible ranking functions, every boundary between a satisfactory region and an unsatisfactory region must comprise \rankswitch functions.  
%Identifying the \rankswitches between the pairs of items in $\mathcal{D}$ is a key idea here.

\subsection{Offline processing}\label{subsec:2doffline}

The goal of offline processing is to identify and index the satisfactory functions in a way that allows efficient answering of online queries.
Following the example in Figure~\ref{fig:2dexplain}, we propose a {\em ray sweeping} algorithm for identifying satisfactory functions in 2D.

\begin{figure*}[!tb] 
    \begin{minipage}[t]{0.2\linewidth}
        \centering
        \vspace{-25mm}
        \begin{tabular}{|l|c|c|}
               \hline
                       $t_1$&1&3.5 \\ \hline
                       $t_2$&1.5&3.1 \\ \hline
                       $t_3$&1.91&2.3 \\ \hline
                       $t_4$&2.3&1.8 \\ \hline
                       $t_5$&3.2&0.9 \\ \hline
        \end{tabular}
        \vspace{3mm}\caption{A 2D dataset}
        \label{fig:dual1}
    \end{minipage}
    %\hspace{1mm}
    \begin{minipage}[t]{0.26\linewidth}
        %\centering
        \includegraphics[width = \textwidth]{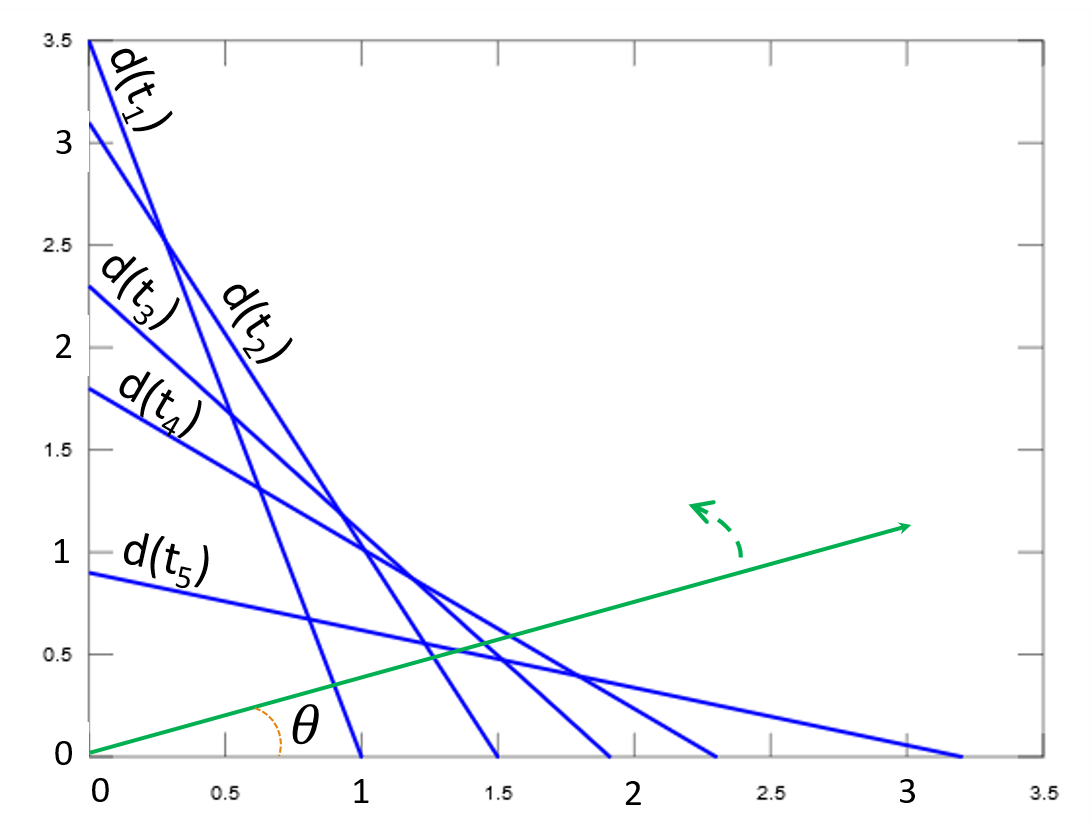}
        \vspace{-6mm}\caption{Dual presentation of Fig.~\ref{fig:dual1}}
        \label{fig:dual2}
    \end{minipage}
    %\hspace{1mm}
    \begin{minipage}[t]{0.26\linewidth}
        %\centering
        \includegraphics[width = \textwidth]{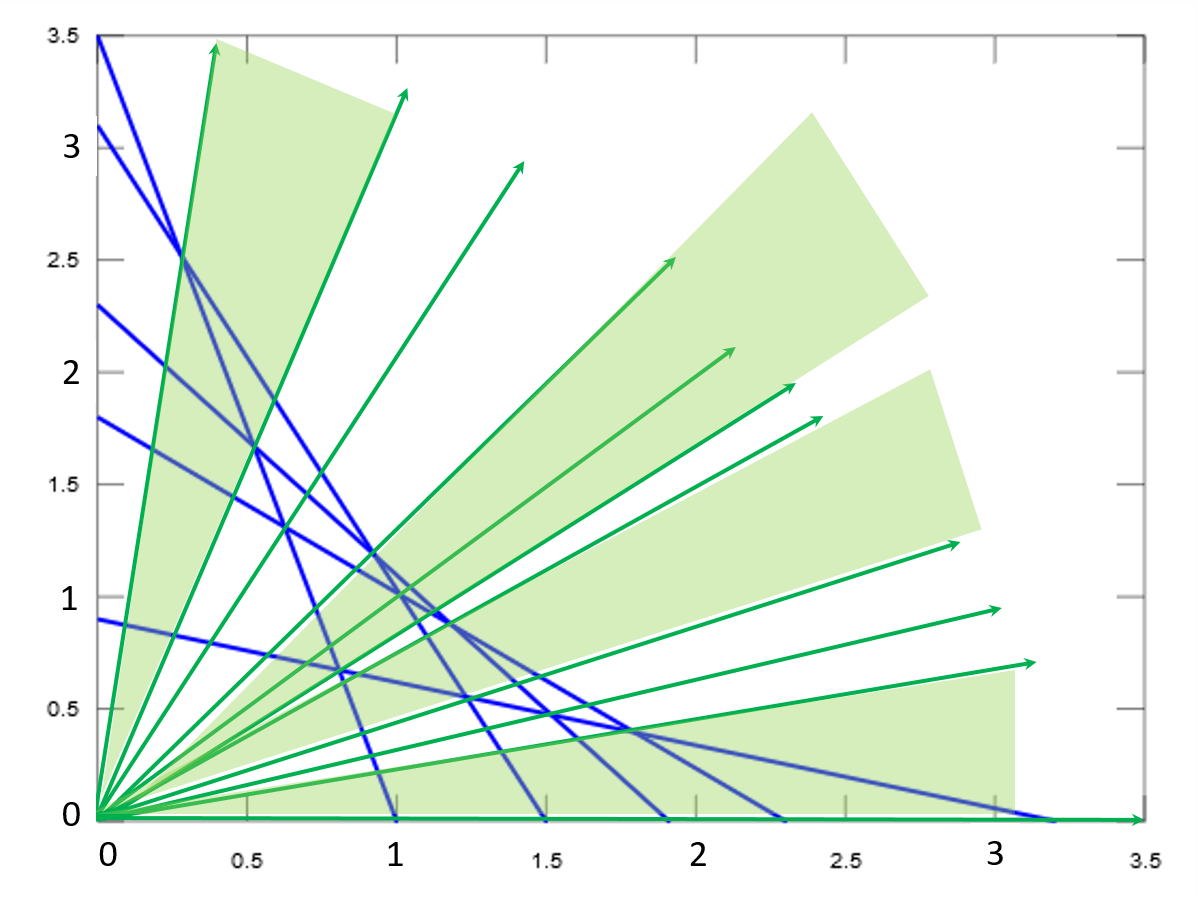}
        \vspace{-6mm}\caption{Satisfactory sectors}
        \label{fig:dual3}
    \end{minipage}
    %\hspace{1mm}
    \begin{minipage}[t]{0.26\linewidth}
        %\centering
        \includegraphics[width = \textwidth]{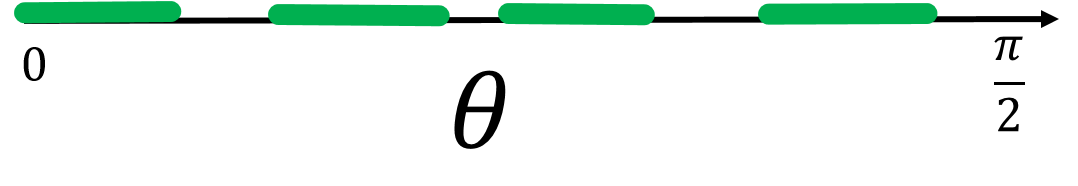}
        \vspace{-6mm}\caption{Satisfactory regions}
        \label{fig:dual4}
    \end{minipage}
\end{figure*}

To identify the \rankswitches of pairs of items, we transform items into a dual space~\cite{edelsbrunner}, where every item $t$ is transformed into the line $\mathsf{d}(t)$, as follows:
\begin{align}\label{eq:2ddual}
\mathsf{d}(t): t[1].x + t[2].y = 1
\end{align}
%Let the line be the dual transformation of the item $t\in \mathcal{D}$.
The ordering of the items based on a function $f$ with the weight vector $\{ w_1, w_2\}$ is the ordering of the intersections of the lines $\mathsf{d}(t)$ with the ray starting from the origin and passing through the point $\langle w_1, w_2\rangle$\footnote{\small This is because the line showing the contour of $t$, transforms to the intersection point of $\mathsf{d}(t)$ and the ray of $f$. %\julia{We no longer talk about contours, is this footnote necessary?}\abol{this footnote is to explain why an item transformed to a line, but the function is still a ray.}\julia{I still don't understand this.  What's a contour?}\abol{contour of a point are the set of points with the same score. we have defined in in section 2 -- p3c1, bottom of the page}
}.
For example, Figure~\ref{fig:dual2} shows the dual transformation (using Eqaution~\ref{eq:2ddual}) of the 2D dataset provided in Figure~\ref{fig:dual1}.
Therefore, the \rankswitch of a pair $t_i$ and $t_j$ is the intersection of $\mathsf{d}(t_i)$ and $\mathsf{d}(t_j)$.
For example, in Figure~\ref{fig:dual2}, the \rankswitch of $t_1$ and $t_2$ is the top-left intersection (of lines $\mathsf{d}(t_1)$ and $\mathsf{d}(t_2)$). %\julia{Figures 4, 5 and 6 need larger font on the axes.}

Using Equation~\ref{eq:2ddual}, the intersection of the lines $\mathsf{d}(t_i)$ and $\mathsf{d}(t_j)$ can be computed by solving the following system of equations:
\[
    \times_{\mathsf{d}(t_i),\mathsf{d}(t_j)}: \left\{
                \begin{array}{ll}
                  t_i[1]x + t_i[2]y=1\\
                  t_j[1]x + t_j[2]y=1
                \end{array}
              \right.
  \]
The \rankswitch is the origin-starting ray with the angle:
\begin{align}\label{eq:2dintersect}
\nonumber \Rightarrow &x = (1-\frac{t_i[2]}{t_j[2]})/(t_i[1]-\frac{t_j[1]t_i[2]}{t_j[2]}) \\
\nonumber \Rightarrow &y = \frac{1-t_i[1]x}{t_i[2]} \\
          \Rightarrow &\theta = \arctan(y/x)
\end{align}

\begin{algorithm}[!h]
\caption{{\bf \twodray}\\
         {\bf Input:} dataset $\mathcal{D}$ and fairness oracle $\mathcal{O}$ \\
         {\bf Output:} sorted satisfactory regions $S$
        }
\begin{algorithmic}[1]
\label{alg:2dray}
    \STATE $\Theta = \{ 0\}$
    \FOR{i = 1 to $n$-1}
        \FOR{j = i+1 to $n$}
            \STATE {\bf if} $t_i$ or $t_j$ dominates the other {\bf then continue}
            \STATE $\theta_{ij}$ = Angle of \rankswitch of $t_i$ and $t_j$ (Eq.~\ref{eq:2dintersect})
            \STATE add $(\theta_{ij},t_i,t_j)$ to $\Theta$
        \ENDFOR
    \ENDFOR
    \STATE sort $\Theta$\label{line:2d-sort}
    \STATE $\Omega$ = sort $\{ t_i\in\mathcal{D} \}$ on $x$-axis, $i=1$
    \WHILE{$i<|\Theta|$ and $\mathcal{O}(\Omega)$ = False}
        \STATE $i= i+1$
        \STATE $(\theta,a,b) = \Theta[i]$
        \STATE swap $a$ and $b$ in $\Omega$
    \ENDWHILE
    \STATE $S = [\langle \Theta[i],0\rangle ]$, flag = True
    \FOR { $j= i+1$ to $|\Theta|-1$}
        \STATE $(\theta,a,b) = \Theta[j]$
        \STATE swap $a$ and $b$ in $\Omega$
        \STATE sign = $\mathcal{O}(\Omega)$
        \STATE {\bf if} flag = True and sign = False {\bf then} append($S$, $\langle \Theta[j],1\rangle$)
        \STATE {\bf else if} flag = False and sign = True {\bf then} append($S$, $\langle \Theta[j],0\rangle$)
        \STATE flag = sign
    \ENDFOR
    \STATE {\bf if} flag = True {\bf then} append($S$, $\langle \pi/2,1\rangle$)
    \STATE {\bf return} $S$
\end{algorithmic}
\end{algorithm}

%\julia{Need a segway, this is abrupt.}
Now, we use the \rankswitches to design the ray sweeping algorithm \twodray, presented in Algorithm~\ref{alg:2dray}: the algorithm first computes the \rankswitch between the pairs of items that do not dominate each other\footnote{\small $t$ dominates $t'$ if $\forall i\in [1,d]$, $t[i]\geq t'[i]$ and $\exists j\in [1,d]$ such that $t[j] > t'[j]$.~\cite{asudeh2016discovering}} %with different types\footnote{\small Please note that the ordering between the items with the same types does not affect the output of the fairness oracle.}, 
using Equation~\ref{eq:2dintersect}, and adds them, in addition to the angle $0$, to a list. 
%\julia{No need to look at all pairs of items, only at those where neither dominates the other.  I think this can be done in less than $O(n^2)$, and, if so, would reduce the over-all complexity. Let's brush up on Bentley~\cite{DBLP:journals/cacm/Bentley80}.}
%\abol{it does not change the complexity. Example: consider a dataset that all the items are in skyline.}
Next, Line~\ref{line:2d-sort} sorts the angles in an ascending order.
Then it orders the items in $\mathcal{D}$ based on the $x$-axis (angle $0$) and gradually updates the ordered list ($\Omega$) as it sweeps the ray toward the  $y$-axis (angle $\pi/2$), by changing the order of pairs of items in their \rankswitches.
Upon finding a satisfactory sector, the algorithm continues attaching the neighboring sectors as long as those are still satisfactory, to generate a satisfactory region.

Algorithm~\ref{alg:2dray} stores the borders of the satisfactory regions in $S$ as pairs $\langle \theta, 0/1\rangle$, where $\langle \theta, 0\rangle$ represents that $\theta$ is the start of a satisfactory region, while $\langle \theta, 1\rangle$ represents that $\theta$ is the end of the region.  Consider Figure~\ref{fig:dual3} and suppose that the green sectors are labeled as satisfactory by the fairness oracle.  Figure~\ref{fig:dual4} shows the satisfactory regions produced by Algorithm~\ref{alg:2dray}. One can see the third from the left satisfactory region is the union of two neighboring satisfactory sectors in Figure~\ref{fig:dual3}.

\begin{theorem}\label{th:2dray}
Algorithm~\ref{alg:2dray} has time complexity $O(n^2 (\log n + \mathbb{O}_n))$, where $\mathbb{O}_n$ is the time complexity of $\mathcal{O}$ for input of size of $n$.
\end{theorem}
\techrep{
\begin{proof}
The proof for this theorem is straightforward, following the number of \rankswitchesnospace.
Since every pair of items in 2D has at most one \rankswitchnospace, the total number of \rankswitches is in $O(n^2)$. Sorting the \rankswitches in Line~\ref{line:2d-sort} is in $O(n^2\log n)$.
Sorting the items along the x-axis is in $O(n\log n)$. 
Then in lines 11 to 24, the algorithm gradually updates the ranked list as it moves from each sector to the next one.
For each of the sectors, it calls the oracle once to check if it is satisfactory. This is in $O(n^2 \mathbb{O}_n)$.
Therefore \twodray is in $O(n^2 (\log n + \mathbb{O}_n))$.
\end{proof}
}
\submit{Proof is given in the technical report~\cite{techreport}.}

\submit{\vspace{1cm}}
\subsection{Online processing}\label{subsec:2donline}
Having the sorted list of 2D satisfactory regions constructed in the offline phase allows us to design an efficient algorithm for online answering of the users' queries.  Recall that a query is a proposed set of weights for a linear ranking function.  Our task is to determine whether these weights result in a fair ranking, and to suggest weight modifications if they do not.

Online processing is implemented by Algorithm~\ref{alg:2donline} that, given a function $f$, applies binary search on the sorted list of satisfactory regions. If $f$ falls within a satisfactory region, the algorithm returns $f$, otherwise it returns the satisfactory border closest to $f$.  %Note that the {\em fairness oracle is not invoked}. 

\begin{algorithm}[!h]
\caption{{\bf \twodonline}\\
         {\bf Input:} sorted satisfactory regions $S$, function $f: \{ w_1, w_2\}$ \\
         {\bf Output:} weight vector $\{ w'_1, w'_2\}$
        }
\begin{algorithmic}[1]
\label{alg:2donline}
    \STATE $(r,\theta ) = (\sqrt{w_1^2+ w_2^2}, \arctan\frac{w_2}{w_1})$
    \STATE low = 1, high = $|S|$
    \WHILE{$($high$-$low$) > 1$}
        \STATE mid = (low+high)/2
        \STATE {\bf if} $S[$mid$][1] < \theta$ {\bf then} low = mid
        \STATE {\bf else} high = mid
    \ENDWHILE
    \IF{ $S[$low$][2] = 0$}
        \STATE {\bf return} $\{ w_1, w_2\}$ \texttt{\scriptsize // input vector is satisfactory}
    \ENDIF
    \IF{$(\theta-S[$low$][1]) < (S[$high$][1]-\theta)$}
            \STATE {\bf return} $\{ r\cos (S[$low$][1]), r\sin (S[$low$][1])\}$
        \ENDIF
    \STATE {\bf return} $\{ r\cos (S[$high$][1]), r\sin (S[$high$][1])\}$
\end{algorithmic}
\end{algorithm}

\begin{theorem}
Algorithm~\ref{alg:2donline} has time complexity $O(\log n)$.
\end{theorem}
\techrep{
\begin{proof}
There totally are at most $O(n^2)$ \rankswitches for $n$ items.
Therefore, the size of the sorted list of satisfactory regions in 2D is in $O(n^2)$.
Applying  binary search on this list is $O(\log n)$.
\end{proof}
}
\submit{Proof is given in the technical report~\cite{techreport}.}

%Next, in \S~\ref{sec:md}, we discuss the extension of 2D algorithms for MD (where the ranking functions are defined on more than two attributes) and their drawbacks.
%We will later propose several techniques for resolving these drawbacks in \S~\ref{sec:speedup}.
\section{The Multi-Dimensional Case}\label{sec:md}
In general, more than two attributes may be used for ranking.  We now extend the basic framework introduced in \S~\ref{sec:2d} to handle multi-dimensional cases.  The challenge is that 
regions of interest are no longer simple planar wedges, bounded by two rays at an angle.  Rather, they are high-dimensional objects, with multiple bounding facets.  

To manage the geometry better, we first introduce an angle coordinate system, and show
that \rankswitches form hyperplanes in this system. Identifying and indexing satisfactory regions during offline processing is similar to constructing the {\em arrangement} of these hyperplanes~\cite{edelsbrunner}.
We then propose an exact online algorithm that works based on the indexed satisfactory regions.
%We conclude this section with a complexity analysis of the proposed algorithms.
%and discussing its pitfalls to work as an online answering of the queries.
%Later in \S~\ref{sec:speedup}, we will propose several techniques for resolving these drawbacks.

\begin{figure*}[!tb] 
    \begin{minipage}[t]{0.18\linewidth}
        \centering
        \vspace{-25mm}
        \begin{small}
        \begin{tabular}{|l|c|c|c|}
               \hline
                       $t_1$&1&2&3 \\ \hline
                       $t_2$&2&4&1 \\ \hline
                       $t_3$&5.3&1&6 \\ \hline
                       $t_4$&3&7.2&2 \\ \hline
        \end{tabular}
        \end{small}
        \vspace{8mm}\caption{A 3D dataset}\label{fig:mdhyperpolar1}
    \end{minipage}
    \hspace{1mm}
    \begin{minipage}[t]{0.29\linewidth}
        %\centering
        \includegraphics[width = 0.9\textwidth]{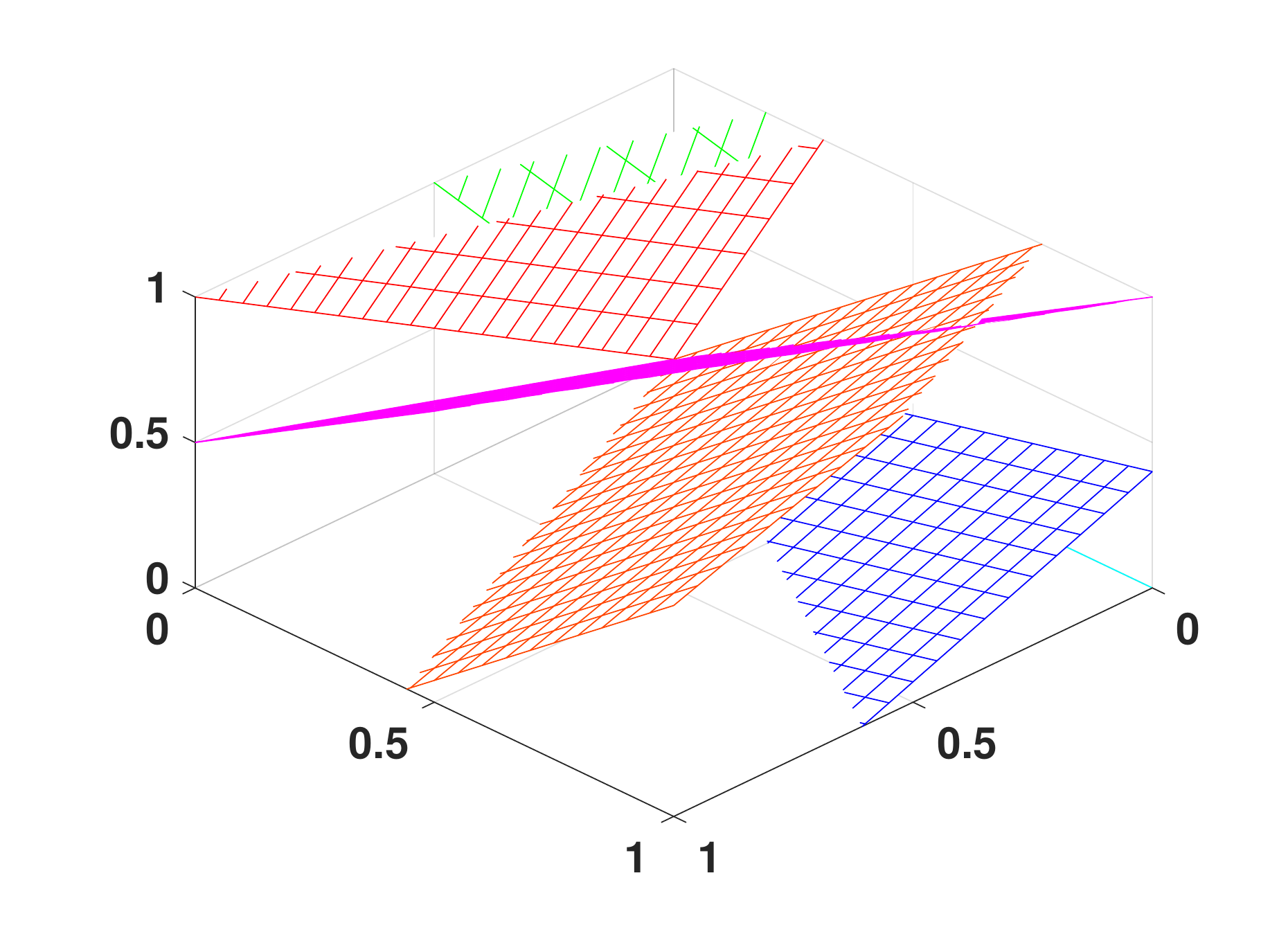}
        \vspace{-3mm}\caption{Ordering exchanges for Fig.~\ref{fig:mdhyperpolar1}}\label{fig:mdhyperpolar2}
    \end{minipage}
    \hspace{1mm}
    \begin{minipage}[t]{0.2\linewidth}
        %\centering
        \includegraphics[width = 0.95\textwidth]{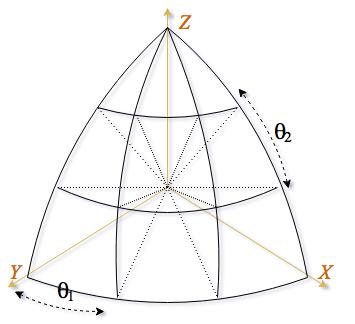}
        \vspace{-3mm}\caption{Angles in $\mathbb{R}^3$}\label{fig:hd1}
    \end{minipage}
    \hspace{1mm}
    \begin{minipage}[t]{0.31\linewidth}
        %\centering
        \includegraphics[width = 0.95\textwidth]{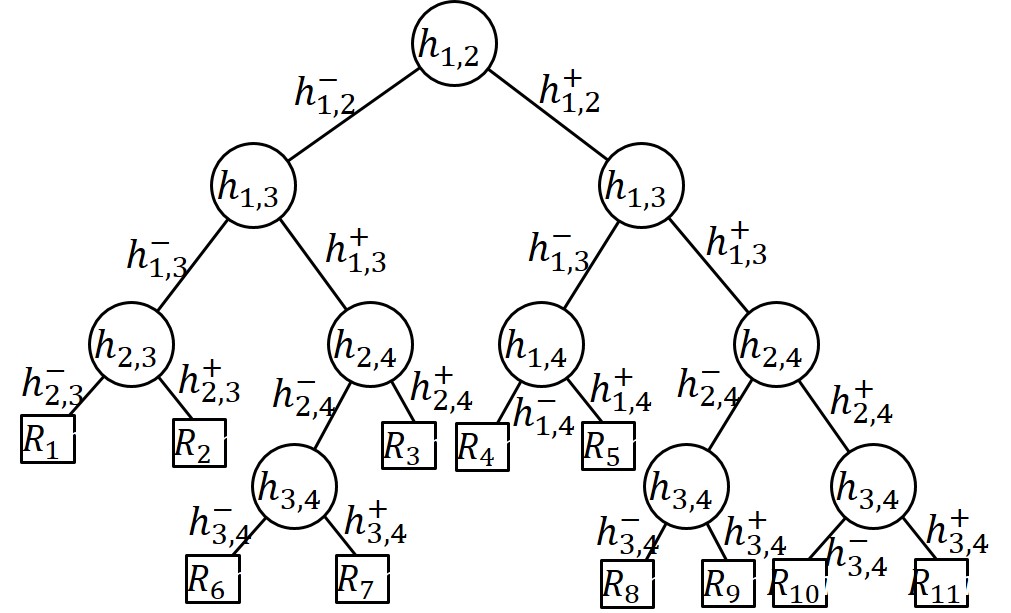}
        \vspace{-3mm}\caption{Arrangement tree example}\label{fig:arrangementtree}
    \end{minipage}
\end{figure*}

\subsection{Ordering exchange in angle coordinates}
Consider function $f$ with weight vector $\vec{w} = \{w_1,w_2,\cdots,w_d \}$.
The score of each tuple $t_i$ based on $f$ is $\Sigma_{k=1}^d w_k t_i[k]$.
For every pair of items $t_i$ and $t_j$, the \rankswitch is the set of functions that give the same score to both items.
As in the previous section, we consider the dual space, transforming item $t$ into a $(d-1)$-dimensional hyperplane in $\mathbb{R}^d$:
\begin{align}\label{eq:2ddual}
\mathsf{d}(t): \sum_{k=1}^d t[k].x_k = 1
\end{align}

For each pair of items $t_i$ and $t_j$, the intersection of $\mathsf{d}(t_i)$ and $\mathsf{d}(t_j)$ is a $(d-2)$ dimensional structure. For instance, in $\mathbb{R}^3$ the dual transformation of every item is a plane and the intersection of two planes is a line.
%Similar to 2D, 
The intersection between $\mathsf{d}(t_i)$ and $\mathsf{d}(t_j)$ can be computed using the following system of equations:
\begin{align}\label{eq:mddualintersect}
    \times_{\mathsf{d}(t_i),\mathsf{d}(t_j)}: \left\{
                \begin{array}{ll}
                  \sum_{k=1}^d t_i[k].x_k = 1\\
                  \sum_{k=1}^d t_j[k].x_k = 1
                \end{array}
              \right.
\end{align}
The set of rays starting from the origin and passing through the points $p\in\times_{\mathsf{d}(t_i),\mathsf{d}(t_j)}$ represents the \rankswitch of $t_i$ and $t_j$.
Hence, the $(d-1)$-dimensional hyperplane defined by $\times_{\mathsf{d}(t_i),\mathsf{d}(t_j)}$ and the origin point (Equation~\ref{eq:mdIntersect}) contains these rays.
\begin{align}\label{eq:mdIntersect}
\sum\limits_{k=1}^{d} (t_i[k] - t_j[k])w_k = 0
\end{align}
For example, consider items $t_1=\{ 1,2,3\}$ and $t_2=\{ 2,4,1\}$ in Figure~\ref{fig:mdhyperpolar1}. Using Equation~\ref{eq:mdIntersect}, the \rankswitch of $t_1$ and $t_2$ is defined by the magenta plane $w_1+2w_2-2w_3 = 0$ in Figure~\ref{fig:mdhyperpolar2}. %\julia{Need larger axis font in this figure.}

As explained in \S~\ref{sec:pre}, linear functions over $d$ attributes (rays in $\mathbb{R}^d$) are identified by $d-1$ angles, each between $0$ and $\pi/2$. For instance, in \S~\ref{sec:2d}, we identify every function in 2D by an angle $\theta\in [0,\pi/2]$.
%Remember from the 2D case that the space of linear functions is identified by an angle $\theta\in [0,\pi/2]$. In 2D we found the angles of the satisfactory ranges and indexed them for the online processing. In $\mathbb{R}^d$, the functions (rays) are defined by $(d-1)$ angles.
Similarly, in multiple dimensions, we identify the functions by their angles.  We now introduce the angle coordinate system for this purpose.

\vspace{2mm}
\noindent{\bf Angle coordinate system:}
Consider the $\mathbb{R}^{d-1}$ coordinate system, where every axis $\theta_i\in [0,\pi/2]$ stands for the angle $\theta_i$ in the polar representation of points in $\mathbb{R}^d$.
Every function (ray in $\mathbb{R}^d$) is represented by the point $\langle \theta_1, \theta_2, \cdots, \theta_{d-1} \rangle$ in the angle coordinate system.
For example, 
as depicted in Figure~\ref{fig:hd1}, a function $f$ in $\mathbb{R}^3$ is the combination of two angles $\theta_1$ and $\theta_2$, each over the range $[0,\pi/2]$.
%Hence, the angle coordinate system for $\mathbb{R}^3$, shown in Figure~\ref{fig:hd2}, has two axes $\theta_1$ and $\theta_2$, each defined over the range $[0,\frac{\pi}{2}]$.

%We study the \rankswitches of the items in the angle coordinate system.
%The angles of the rays (starting from the origin) in the hyperplane defined by Equation~\ref{eq:mdIntersect}, i.e., 
Following Equation~\ref{eq:mdIntersect}, the \rankswitch of a pair of items forms a $(d-2)$-dimensional hyperplane in the angle coordinate system. For example, in 3D, the \rankswitch of $t_i$ and $t_j$ forms a line. We use $h_{i,j}$ to refer to the \rankswitch of $t_i$ and $t_j$ in the angle coordinate system.

\begin{algorithm}[!h]
\caption{{\bf \hyperpolar}\\
         {\bf Input:} items $t_i$ and $t_j$\\
         {\bf Output:} \rankswitch $h_{i,j}$
        }
\begin{algorithmic}[1]
\label{alg:hyperpolar}
    \STATE $V = [ t_i[k]-t_j[k] ~,~ \forall 1\leq k \leq d ]$ \texttt{\scriptsize  // Equation~\ref{eq:mdIntersect}}
    \STATE $\Theta= \{ \}$ %\\ \texttt{\scriptsize  // Find d-1 linearly independent points satisfying Equation~\ref{eq:mdIntersect}}
    \STATE $p =$ $d-1$ linearly independent points satisfying Equation~\ref{eq:mdIntersect} 
    \FOR{$k=1$ to $d-1$}
        %\STATE sum$=0$
        %\FOR{$l=1$ to $d-1$}
        %    \STATE $p[l] = $primes$[l]^k$
        %    \STATE sum$ = $sum$+V[l].p[l]$
        %\ENDFOR
        %\STATE $p[d] = -$sum$/v[d]$
        \STATE $(r,\theta)$ = {\bf ToPolar}($p[k]$)
        \STATE add $\theta$ to $\Theta$
    \ENDFOR \\ \texttt{\scriptsize  // Find the hyperplane containing the points in $\Theta$}
    \STATE $\iota=[1,1,\cdots ,1]$
    \STATE {\bf return} $\Theta^{-1}\times \iota$
\end{algorithmic}
\end{algorithm}

Before we can construct satisfactory regions, we first need to compute \rankswitches in the angle coordinate system.
Algorithm~\ref{alg:hyperpolar} computes $h_{i,j}$ for a given pair of items $t_i$ and $t_j$.
The algorithm uses %using the prime numbers,
$(d-1)$ linearly independent points in the hyperplane of Equation~\ref{eq:mdIntersect}, %In order to do so, it uses the prime numbers to, starting from a point other than origin, independently scale each dimension to get new points.
and finds the angles of the ray from the origin through each of the points, using their polar representations.
%\julia{I don't understand this.  With which non-zero point do we start, with any?  And what are ``the other points''?}
To find the points, one can start with an arbitrary non-zero point on the plane and scale each dimension independently to get the other points.
%\gautam{How do we compute these linearly independent points?}
After this step, each row of the $(d-1) \times (d-1)$ matrix $\Theta$ shows a point in the angle coordinate system.
\hyperpolar represents hyperplanes as $\sum_{k=1}^{d-1} h_{i,j}[k]\theta_k=1$.
Since all $(d-1)$ points in $\Theta$ fall in $h_{i,j}$, this forms a linear system of equations $\Theta\times h_{i,j} = \iota$, where $\iota$ is the unit vector of size $(d-1)$.
Solving this system of equations, we get $h_{i,j} = \Theta^{-1}\times \iota$.
Given that computing $\Theta^{-1}$ is the bottleneck in Algorithm~\ref{alg:hyperpolar}, it is easy to see that \hyperpolar is in $O(d^3)$, which is $O(1)$ for a fixed $d$.

\noindent
%Having the \rankswitches transformed into the angle coordinate system, next we discuss their arrangement for finding the satisfactory regions.

\subsection{Construction of satisfactory regions}\label{subsec:arrangementconstruction}
The construction of satisfactory regions relates to the arrangement~\cite{edelsbrunner} of \rankswitch hyperplanes in the angle coordinate system.
Consider the arrangement of $h_{i,j}$, $\forall t_i$, $t_j\in \mathcal{D}$.
%For example, Figure~\ref{fig:hd2} shows an arrangement of \rankswitches between 4 items when $d=3$.
Items $t_i$ and $t_j$ switch order on the two sides of $h_{i,j}$,
%As discussed previously, for every pair of items $t_i$ and $t_j$, the ordering at each side of $h_{i,j}$ is different.
while inside each convex region in the arrangement their relative ordering does not change.
In the following, we construct all convex regions in the arrangement and check if the ordering inside each is satisfactory.

A convex region %in the arrangement 
is defined as the intersection of a set of half-spaces~\cite{edelsbrunner}. Every hyperplane $h$ divides the space into two half-spaces $h^+$ and $h^-$. In our problem, 
 the ordering between $t_i$ and $t_j$ switches for each hyperplane $h_{i,j}$, moving from $h_{i,j}^+$ to $h_{i,j}^-$.
%considering each half-space as the \rankswitch of a pair $t_i$ and $t_j$, the ordering between $t_i$ and $t_j$ switches,  from $h_{i,j}^+$ to $h_{i,j}^-$.

Inspired by the algorithm proposed in~\cite{edelsbrunner}, we develop an incremental algorithm for discovering the convex regions in the arrangement.
%We propose an incremental algorithm similar to the one proposed for constructing the arrangement of hyperplanes in~\cite{edelsbrunner}. However, rather than
Intuitively, Algorithm~\ref{alg:mdsatregions} adds the hyperplanes one after the other to the arrangement. At every iteration, it finds the set of regions in the arrangement with which the new hyperplane intersects.
Recall that $h_{i,j}$ is in the form of $\sum_{k=1}^{d-1} h_{i,j}[k] \theta_k = 1$.
Hence, the half-space  $h_{i,j}^+$ can be considered as the constraint $\sum_{k=1}^{d-1} h_{i,j}[k] \theta_k \geq 1$ and $h_{i,j}^-$ as $\sum_{k=1}^{d-1} h_{i,j}[k] \theta_k \leq 1$.
The set of points inside a convex region $R=\{ (h_{R1},+/-), (h_{R2},+/-), \cdots\}$ satisfy constraints $\sigma_R$ as defined in Equation~\ref{eq:CofR}.

\begin{align} \label{eq:CofR}
\sigma_R: \left\{
                \begin{array}{ll}
                  \forall \mbox{ half-space} (h',+)\in R,~ \sum_{k=1}^{d-1} h'[k] \theta_k \geq 1\\
                  \forall \mbox{ half-space} (h',-)\in R,~ \sum_{k=1}^{d-1} h'[k] \theta_k \leq 1
                \end{array}
              \right.
\end{align}

Using Equation~\ref{eq:CofR}, a hyperplane $h$ intersects with a convex region $R$ if there exists a point $p\in h$ such that the constraints in $\sigma_R$ are satisfied. The existence of such a point can be determined using linear programming (LP).  
If the new hyperplane intersects with $R$, Algorithm~\ref{alg:mdsatregions} breaks it down into two convex regions that represent the intersections of $R$ with half-spaces $h^+$ and $h^-$.

Having constructed the arrangement, % (i.e., adding all hyperplanes to the arrangement)
Algorithm~\ref{alg:mdsatregions} finds, using linear programming, a point $\theta$ that satisfies $\sigma_R$, and uses $\theta$ to check if region $R$ is satisfactory.  If $R$ is not satisfactory, it is removed from the set of satisfactory regions $\mathcal{R}$.

\begin{algorithm}[!h]
\caption{{\bf \arrangement}\\
         {\bf Input:}  dataset $\mathcal{D}$ and fairness oracle $\mathcal{O}$ \\
         {\bf Output:} satisfactory regions $\mathcal{R}$
        }
\begin{algorithmic}[1]
\label{alg:mdsatregions}
    \STATE $H=\{\}$ \\ \texttt{\scriptsize  // construct \rankswitches in angle coordinates}
    \FOR{$i = 1$ to $n-1$}
        \FOR{$j = i+1$ to $n$}
                \STATE {\bf if} $t_i$ or $t_j$ dominates the other {\bf then continue}
                \STATE add \hyperpolar($t_i$ ,$t_j$) to $H$
        \ENDFOR
    \ENDFOR
    \STATE $\mathcal{R}=\{~\{(H[1],+)\},\{(H[1],-)\}~\}$ \\ \texttt{\scriptsize  // add hyperplanes incrementally to the arrangement}
    \FOR{$h\in (H\backslash \{H[1]\})$}
    		\STATE $\ell_\mathcal{R} = |\mathcal{R}|$
            \FOR{$i = 1$ to $\ell_\mathcal{R}$}
                \IF{$\exists p\in h$ s.t. $\sigma_{\mathcal{R}[i]}$}
                    \STATE $R' = \mathcal{R}[i]$
                    \STATE append $\mathcal{R}[i]$ by $(h,+)$
                    \STATE append $R'$ by $(h,-)$
                    \STATE add $R'$ to $\mathcal{R}$
                \ENDIF
            \ENDFOR
    \ENDFOR \\ \texttt{\scriptsize  // remove the unsatisfactory regions}
    \FOR{$R\in\mathcal{R}$}
            \STATE $\theta = $ a point that $\sigma_R$ is satisfied
            \STATE $\vec{w}$ = {\bf ToCartesian}(1,$\theta$)
            \IF{ $\mathcal{O}($OrderBy$_{f_{\vec{w}}}(\mathcal{D}))=$ False}
                \STATE remove $R$ from $\mathcal{R}$
            \ENDIF
    \ENDFOR
    \STATE {\bf return} $\mathcal{R}$
\end{algorithmic}
\end{algorithm}

\begin{theorem}\label{th:3}
For a fixed number of dimensions, the time complexity of Algorithm~\ref{alg:mdsatregions} is $O\big(n^{2d-1} (n\, Lp(n^2) + \mathbb{O}_n\log n )\big)$, where $Lp(n^2)$ is the time of solving a linear programming problem of $n^2$ constrains and a fixed number of variables and $\mathbb{O}_n$ is the time complexity of $\mathcal{O}$ for an input of size $n$.
\end{theorem}
\submit{Proof is given in the technical report~\cite{techreport}.
\vspace{2mm}}
\techrep{
\begin{proof}
%\gautam{This theorem is stated strangely. On one hand, you say the number of dimensions is constant. On the other hand, $d$ also shows up in the complexity analysis.} 
Lines 2 to 6 of \arrangement construct $h_{i,j}$ for each pair of the items $t_i$ and $t_j$ (in the dual space). Since Algorithm~\ref{alg:hyperpolar} 
has a constant complexity for a fixed number of dimensions, constructing the \rankswitches in the angle coordinate system is in $O(n^2)$.
%\gautam{Why is Algorithm~\ref{alg:hyperpolar} in $O(d^2)$? Why don't you provide the analysis just after the algorithm is described?}
The next step of the algorithm is constructing the arrangement of hyperplanes. 
Using results from combinatorial geometry, the complexity of the arrangement of $n^2$ hyperplanes in $\mathbb{R}^{d-1}$ is $O(n^{2(d-1)})$~\cite{edelsbrunner}.
The bottleneck in Algorithm~\ref{alg:mdsatregions} is the construction of the arrangement: at iteration $i$, add the $i^{th}$ hyperplane to the arrangement.
To do so, identify the set of regions with which the current hyperplane intersects, by 
applying a linear scan over the set of regions to find intersections. Furthermore, for each region, Algorithm~\ref{alg:mdsatregions} solves an LP with $i^2$ constraints over a fixed number of variables.
The number of regions at iteration $i$ is $i^{2(d-1)}$.
Thus the total cost is:
$$O(\sum\limits_{i=1}^{n^2} (i^{2(d-1)} Lp(i^2))) \leq O(n^{2d}Lp(n^2))$$
After constructing the arrangement, the algorithm removes the unsatisfactory regions from $\mathcal{R}$. To do so, for each region, it chooses a function inside the regions, orders the items based on it, and calls the oracle to check if it is satisfactory.
There are $O(n^{2(d-1)})$ regions in the arrangement and ordering the items in each region is in $O(n\log n)$. Hence, this step is in $O(n^{2d-1}\log n \, \mathbb{O}_n)$.
The time complexity of the algorithm, therefore, is: $$O\big(n^{2d-1} (n\, Lp(n^2) + \mathbb{O}_n\log n )\big)$$
\end{proof}
}
%Hence, the upperbound cost of each iteration is equal to the number of regions in that iteration times to the cost of solving the LP problems.
%Therefore to total complexity of Algorithm~\ref{alg:mdsatregions} is $O(n^{4(d-1)}Lp(n^2,d-1))$.

%Most of the effort in Algorithm~\ref{alg:mdsatregions} is in the LP on line 10, which has to be solved once for each region in order to add a new hyperplane.
To add a new hyperplane, Algorithm~\ref{alg:mdsatregions} checks the intersection of every region with the hyperplane. But in practice most regions do not intersect with it.
In the following, we
%use an observation to skip this linear scan.
%Consider a region $R$ at step $i$. Adding the new hyperplanes to the arrangement after step $i$, {\em partitions} $R$ into new regions.
%As a result, while adding a hyperplane $h_{i+j}$, if it does not intersect with $R$, it will not intersect with any of regions inside it.
%Following this observation, we 
define the {\em arrangement tree}, which  keeps tracks of the space partitioning in a hierarchical manner, and can quickly rule out many regions.  While this does not change the asymptotic worst case complexity, we find that it greatly helps in practice, as we will illustrate experimentally in \S~\ref{subsec:exp-performance}.

\vspace{2mm}
\noindent
{\bf Arrangement tree:}
Consider a binary tree where every vertex $v$ is associated with a hyperplane $h_i$, while its left and right edges refer to $h_i^-$ and $h_i^+$, respectively.
Every vertex of the tree corresponds to a region $R$ that is the set of half-spaces specified by the edges from the root to it.
As a result, the left (resp. right) child of $v$ shows the regions in $R$ that fall in $h^-$ (resp. $h^+$).\\
Figure~\ref{fig:arrangementtree} shows a sample arrangement tree for a set of $6$ hyperplanes $\{h_{1,2},h_{1,3},h_{1,4},h_{2,3},h_{2,4},h_{3,4}\}$.
The leaves of the tree are the regions of the arrangement. The region $R_3$, for example, is the intersection of the half-spaces $\{h_{1,2}^-,h_{1,3}^+,h_{2,4}^+\}$.
In this figure, 
consider the left child of the root. Let us assume that a new hyperplane $h$ does not intersect with the right child of this node, i.e., it does not intersect with the region $\{h_{1,2}^-,h_{1,3}^+\}$. Then we can prune the whole subtree and skip checking the intersection of $h$ with the regions $R_3$, $R_6$, and $R_7$, because all these regions are inside the region $\{h_{1,2}^-,h_{1,3}^+\}$.

\techrep{Algorithm~\ref{alg:harrange} shows the recursive algorithm for adding a hyperplane to an arrangement, using the arrangement tree.}
\submit{The pseudocode of \harrangementnsp, the recursive algorithm for adding a hyperplane to an arrangement, using the arrangement tree is provided in the technical report~\cite{techreport}.}
It replaces the lines $9$ to $18$ in Algorithm~\ref{alg:mdsatregions}.

\techrep{
\begin{algorithm}[!h]
\caption{{\bf \harrangement}\\
         {\bf Input:} arrangement tree $T$, the hyperplane $h$, the constraints path to root $\sigma$}
\begin{algorithmic}[1]
\label{alg:harrange}
	\IF{$T$ is null}
		\STATE $T=$ new {\bf ArrangementTree}($h$)
		\STATE {\bf return}
	\ENDIF
    \STATE $\sigma_l = \sigma\cup \{\sum_{k=1}^{d-1} T.h[k] \theta_k \leq 1 \}$
    \STATE $\sigma_r = \sigma\cup \{\sum_{k=1}^{d-1} T.h[k] \theta_k \geq 1 \}$
    \STATE {\bf if} $h$ passes through $\sigma_l$ {\bf then} \harrangement($T$.left,$h$,$\sigma_l$)
    \STATE {\bf if} $h$ passes through $\sigma_r$ {\bf then} \harrangement($T$.right,$h$,$\sigma_r$)
%    \STATE {\bf return} leaves($T$)
\end{algorithmic}
\end{algorithm}
}

\subsection{Online processing} % - \mdbaseline}
Thus far in this section, we studied how to preprocess the data and construct satisfactory regions in multiple dimensions.
Next, given a query (a function $f$) and the satisfactory regions, our objective is to find the closest satisfactory function $f'$ to $f$. 
To do so, \mdbaseline solves a non-linear programming problem for each  satisfactory region to find the closest point of the region to $f$. It then returns the function with the minimum angle distance with $f$.
%\gautam{Can you be a bit more specific about the type of nonlinear optimization? Seems to me that the constraints are linear, while the objective function of minimizing angle distance is nonlinear, right?}

\begin{algorithm}[!h]
\caption{{\bf \mdbaseline}\\
         {\bf Input:} Satisfactory regions $\mathcal{R}$, dataset $\mathcal{D}$, fairness oracle $\mathcal{O}$, function  $f: \vec{w}$ \\
         {\bf Output:} the satisfactory weight vector $\vec{w'}$
        }
\begin{algorithmic}[1]
\label{alg:mdbaseline}
    \IF{$\mathcal{O}($OrderBy$_{f_{\vec{w}}}(\mathcal{D}))=$ True}
        \STATE {\bf return $\vec{w}$}
    \ENDIF
    \STATE $(r,\Theta^{(i)})=$ {\bf ToPolar}$(\vec{w})$
    \STATE mindist=$\infty$
    \FOR{$R\in\mathcal{R}$}
        \STATE (dist,$\Theta^{(j)}$) = the minimum $\theta_{i,j}$ such that $C_R$ is satisfied \techrep{\texttt{\scriptsize // based on Equation~\ref{eq:raydistance}} }
        \IF{dist$<$mindist}
            \STATE $\Theta^o = \Theta^{(j)}$, mindist$=$dist
        \ENDIF
    \ENDFOR
    \STATE {\bf return} {\bf ToCartesian}(1,$\Theta^o$)
\end{algorithmic}
\end{algorithm}

\begin{theorem}
For a constant number of dimensions, the time complexity of Algorithm~\ref{alg:mdbaseline} is $O(n^{2(d-1)}NLp(n^2))$, where $NLp(n^2)$ is the time for solving a non-linear programming problem of $n^2$ constraints and a fixed number of variables.
\end{theorem}
\submit{Proof is given in the technical report~\cite{techreport}.}
\techrep{
\begin{proof}
%\gautam{Again, same comment as earlier. You assume constant $d$, yet $d$ appears in the analysis.}
Given that the upper-bound on the total number of satisfactory regions, is $O(n^{2(d-1)})$, the proof is straightforward.
For every satisfactory region, \mdbaseline needs to solve a non-linear programming problem of size $O(n^2)$ constraints over fixed number of variables.
Thus, Algorithm~\ref{alg:mdbaseline} is in $O(n^{2(d-1)}NLp(n^2 ))$.
\end{proof}
}

%\section{Performance Optimizations}\label{sec:speedup}
\vspace{1cm}
\section{Approximation}\label{sec:speedup}
A user developing a scoring function requires an interactive response from the system.
\mdbaseline is not practical for query answering as it needs to solve a non-linear programming problem for each satisfactory region, before answering each query.
In this section, we propose an efficient algorithm for obtaining approximate answers quickly.
%an alternative preprocessing that makes the online answering of the queries possible.
%We start the section by proposing a tree data structure that makes \arrangement faster in practice.
%Then, in the rest of the section, we focus on making the online query answering efficient.
Our approach relies on first partitioning the angle space, based on a user-controlled parameter $N$, into $N$ {\em cells}, where each cell $c$ is a hypercube of $(d-1)$-dimensions.
We conduct the partitioning in a way that the maximum angle distance between every pair of functions in every cell is bounded.
\submit{
Due to the space limitations, please see the details in the technical report~\cite{techreport}.}
\techrep{lease see the details in Appendix~\ref{subsec:anglepartitioning}.}
%it bounds the maximum angle distance between every pair of functions based on a user-specified parameter, based on the value of $N$ (see the technical report~\cite{techreport} for details). \julia{Let's see which details we move to the TR.  A citation may not be appropriate here.}
%\subsection{Approximation: angle space partitioning}\label{subsec:approx}
%Our approach is based on first partitioning the angle space into $N$ {\em cells}, where each cell $c$ is a $d-1$-dimensional hypercube. 
%and answer each query with a satisfactory function that is guaranteed to be less than a specified maximum distance from the optimal solution.
%So far, we discussed how to partition the angle space such that the maximum angle distance between every pair of functions is bounded.
%\julia{The following paragraph is difficult to understand, because ``distance'' is used to mean two different things, if I understand correctly.  In the exact solution, we minimize $dist(f,f')$.  In the approximate solution, is the statement that: $|dist(f,f_c') - dist(f,f')| \leq \epsilon$, and that $\epsilon$ is specified by the user?  Or does the user specify some other parameter, $N$, and then there is a function that takes $N$ to $\epsilon$? In any case, I wouldn't call $\epsilon$ a distance, let's call it a threshold?}
In the preprocessing, we assign a satisfactory function $f_c'$ to every cell $c$ such that, for every function $f$, 
the angle between $f$ and $f_c'$ is within a bounded threshold (based on the value of $N$) from $f$ and its optimal answer.
To do so, in \S~\ref{subsec:cellmarking}, we first identify the cells that intersect with a satisfactory region, and assign the corresponding satisfactory function to each such cell.
Then, in \S~\ref{subsec:coloring}, we assign the cells that are outside of the satisfactory regions to the nearest discovered satisfactory function.
%\gautam{I think you first need to have a small subsection where you carefully describe how the space is partitioned into cells. Have a notation for the user-controlled parameter, and show how the number of cells $N$ is derivable from this parameter. Also give some insight on the geometry of the cells, i.e., they are hypercubes.}
%In the following, we show how to identify and mark the cells that intersect with a satisfactory region.

\subsection{Identifying cells in satisfactory regions}\label{subsec:cellmarking}

\begin{figure*}[!tb] 
    \begin{minipage}[t]{0.2\linewidth}
        %\centering
        \includegraphics[width = \textwidth]{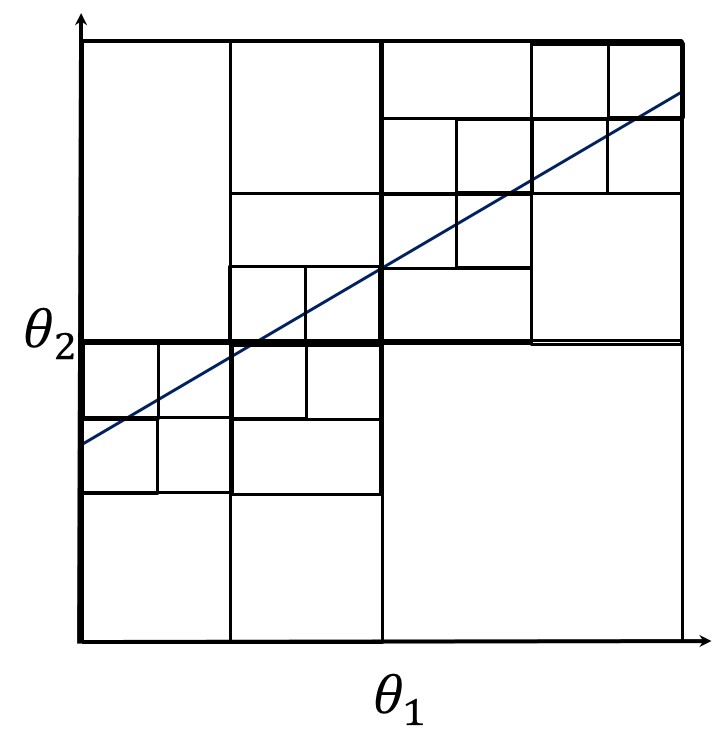}
        \caption{Identifying cells that intersect a hyperplane}\label{fig:quadtree}
    \end{minipage}
    \hspace{1mm}
    \begin{minipage}[t]{0.30\linewidth}
        %\centering
        \includegraphics[width = \textwidth]{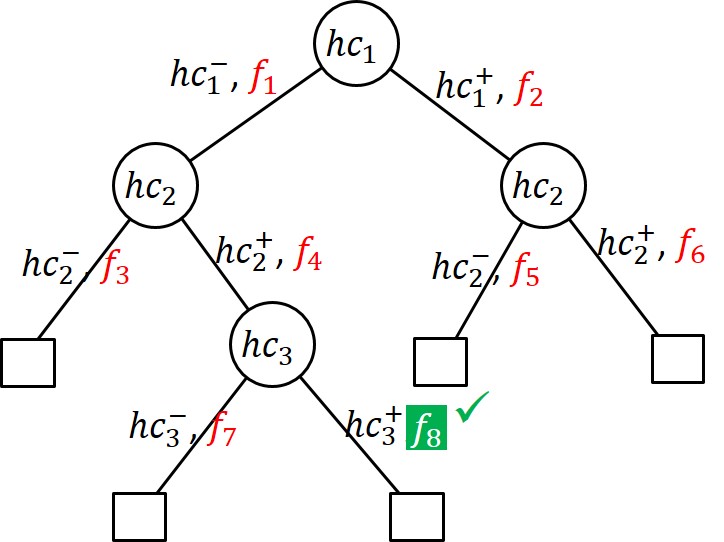}
        \caption{Early stopping when constructing the arrangement of a cell}\label{fig:atc}
    \end{minipage}
    \hspace{1mm}
    \begin{minipage}[t]{0.215\linewidth}
        %\centering
        \includegraphics[width = 0.95\textwidth]{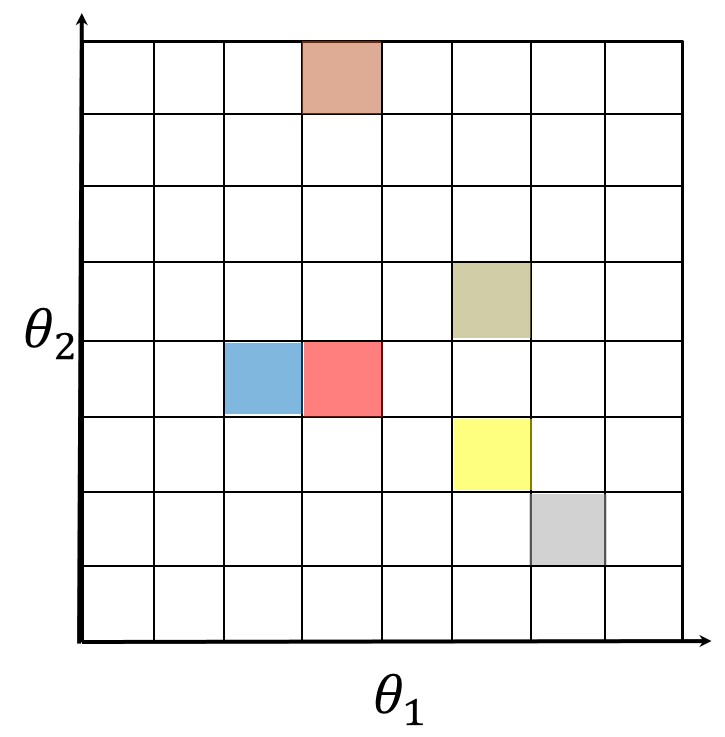}
        \caption{Satisfactory cells example}\label{fig:coloring1}
    \end{minipage}
    \hspace{1mm}
    \begin{minipage}[t]{0.215\linewidth}
        %\centering
        \includegraphics[width = 0.95\textwidth]{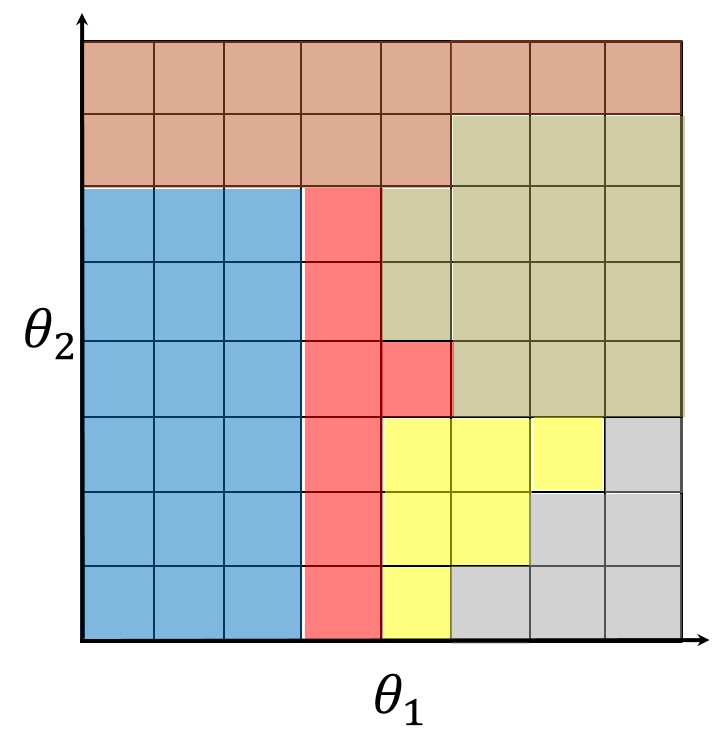}
        \caption{Coloring unsatisfactory cells in Fig.~\ref{fig:coloring1}}\label{fig:coloring2}
    \end{minipage}
\end{figure*}

After partitioning the angle space, our objective here is to find cells in $Cells$, the set of all cells, that intersect with at least one satisfactory region $R\in\mathcal{R}$.
Formally,
\begin{align}\label{eq:satis.cells}
\mathcal{C} = \{c\in Cells ~|~\exists R\in\mathcal{R}\mbox{ s.t. } R\cap c \neq \emptyset \}
\end{align}

A brute force algorithm follows Equation~\ref{eq:satis.cells} literally.
%and first constructs the arrangement and then for every cell in $Cells$ checks if there exists a region $R\in\mathcal{R}$ where their intersection is not empty.
%\julia{What's $N$, the size of the grid?  How is it quantified exactly?}
This algorithm needs to first construct a complete arrangement and then check the intersection of all $N\times |\mathcal{R}|$ pairs of cells and satisfactory regions. Given the potentially large values of $N$ and size of $\mathcal{R}$, this is inefficient.
As discussed in \S~\ref{sec:md}, and experimentally shown in \S~\ref{sec:exp},
the complexity of the arrangement and the running time of Algorithm~\ref{alg:mdsatregions} highly depends on the number of hyperplanes in the arrangement. Even though the first few hyperplanes are quickly added to the arrangement, adding the later hyperplanes is more time consuming. This observation motivates us to limit the construction of the arrangement to subsets of hyperplanes, as opposed to constructing the complete arrangement all at once.
On the other hand, the changes in the ordering in every cell is limited to the hyperplanes passing through it. As a result, for finding out if a cell intersects with a satisfactory region, it is enough to only consider the arrangement of these hyperplanes.

%Here we take a different route for finding the cells that satisfy Equation~\ref{eq:satis.cells}.  As discussed in \S~\ref{sec:md}, and we will experiment in \S~\ref{sec:exp}, the incremental nature of Algorithm~\ref{alg:mdsatregions} makes its performance dependent on the number of hyperplanes in the arrangement. On the other hand, while constructing the complete arrangement is very expensive, finding out if a hyperplane intersects with a cell or not is simple and can be done in constant time.
Given a hyperplane $h$ and a cell $c$, checking if $h$ passes through $c$ is simple, using the ``bottom-left'' ($bl$) and ``top-right'' ($tr$) corners of the cell, i.e., the corners that have the minimum and maximum values of the cell ranges in each dimensions.
%\gautam{Not sure whether I like the terms top-left and bottom-right, they seem to imply 2 dimensions only.}

Recall that \hyperpolar constructs the hyperplane $h$ in the form of $\sum_{k=1}^{d-1} h[k]\theta_k = 1$.
Thus, for every point $p$ in $h^-$, $\sum_{k=1}^{d-1} p_k\theta_k \leq 1$ while for every point $p'$ in $h^+$, $\sum_{k=1}^{d-1} p'_k\theta_k \geq 1$.Therefore, $h$ passes through $c$, {\it iff} 
$\sum_{k=1}^{d-1} bl[k]\theta_k \leq 1$ and $\sum_{k=1}^{d-1} tr[k]\theta_k \geq 1$.

%Assuming that, on average, a small subset of hyperplanes pass through a cell, limiting the arrangement to those hyperplanes reduces the arrangement construction cost significantly. Therefore, rather than applying the expensive arrangement construction for the complete set of hyperplanes, for each cell, we limit the arrangement to the hyperplanes passing through it.

The complete pairwise check between each hyperplane and each cell takes $O(N\times |H|)$ time. Instead, we use the following observation to skip some of the operations: consider a hyperrectangle specified by its bottom-left corner $bl$ and the top-right corner $tr$; also consider a hyperplane $h$ that does not pass through this hyperrectangle.
For every cell $c$ for which its bottom-left dominates $bl$ (for each dimension $i$ its value is greater than or equal to $bl[i]$) and its top-right corner is dominated by $tr$, $h$ does not pass through $c$. 

As a result, for checking the cells that intersect with hyperplane $h$, one can start from the complete angle space, partition the space in a hierarchical manner, and prune the cells inside the hyperrectangles that do not intersect with $h$.
We adopt the {\em quadtree}~\cite{finkel1974quad} data structure for this purpose. %\gautam{reference for quadtree?}
To do so, the recursive
\techrep{Algorithm~\ref{alg:quadtree}}
\submit{algorithm \cellplane}
 iterates over the dimensions in a round robin manner and, at every step, if $h$ passes through the current hyperrectangle, divides it in two equi-size hyperrectangles on the current dimension.
\submit{Please find the pseudo code of \cellplane in the technical report~\cite{techreport}.}

Figure~\ref{fig:quadtree} illustrates \cellplane for finding the cells that intersect with the drawn line $h$. The algorithm prunes all cells in the bottom-right quadrant, since $h$ does not pass through it.
%\gautam{In Figure 11, which direction is increasing  $\theta_1$ (resp. $\theta_2$)? Also, the pseudo-code of Algorithm 7 has several undefined notations.}

\techrep{
\begin{algorithm}[!h]
\caption{{\bf \cellplane}\\
         {\bf Input:} hyperplane $h$, $Cells$, low (indices of bottom-left corner), high (indices of top-right corner), turn (the dimension to divide), and list of hyperplanes for cells $\mathcal{HC}$
        }
\begin{algorithmic}[1]
\label{alg:quadtree}
    \STATE {\bf if} $h$ does not passes through \\ $\cdots$rectangle(bottom-left(low),top-right(high)) {\bf then return}
    \IF {high[turn] = low[turn]}
        \IF {$\forall 1\leq i \leq (d-1):$ low[i]=high[i]}
            \STATE add $h$ to HC[low] and {\bf return}
        \ENDIF
        \STATE {\bf while} high[turn] = low[turn] {\do} turn= (turn+1)mod(d-1)
    \ENDIF
    \STATE mid = low[turn]+high[turn]/2
    \STATE tmp= high[turn]; high[turn] = mid
    \STATE \cellplane ($h$,$Cells$,low,high,(turn+1)mod(d-1),$\mathcal{HC}$)
    \STATE high[turn]=tmp; low[turn] = mid+1
    \STATE \cellplane ($h$,$Cells$,low,high,(turn+1)mod(d-1),$\mathcal{HC}$)
\end{algorithmic}
\end{algorithm}
}

\begin{algorithm}[!h]
\caption{{\bf \cellarrangement}\\
         {\bf Input:} cell $c$, $\mathcal{HC}$
        }
\begin{algorithmic}[1]
\label{alg:cellarrangement}
    \IF{$|\mathcal{HC}[c]|=0$}
        \STATE $p = $ a point inside $c$
        \STATE {\bf if} $\mathcal{O}($OrderBy$_{p}(\mathcal{D}))=$ True {\bf then} Marked[$c$]$=p$
        \STATE {\bf return}
    \ENDIF
    \STATE $p=$ a point in $\mathcal{HC}[c][1]^-\cap c$
    \STATE {\bf if} $\mathcal{O}($OrderBy$_{p}(\mathcal{D}))=$ True {\bf then} Marked[$c$]$=p$; {\bf return}
    \STATE $p=$ a point in $\mathcal{HC}[c][1]^+\cap c$
    \STATE {\bf if} $\mathcal{O}($OrderBy$_{p}(\mathcal{D}))=$ True {\bf then} Marked[$c$]$=p$; {\bf return}
    \STATE $T=$new {\bf ArrangementTree}$(\mathcal{HC}[c][1])$
    \FOR{$h\in\mathcal{HC}[c]\backslash \mathcal{HC}[c][1]$}
        \IF {$p=$\harrangementc($T$,$h$,$c$,null) is not null}
            \STATE Marked[$c$]$=p$; {\bf return}
        \ENDIF
    \ENDFOR
\end{algorithmic}
\end{algorithm}

\begin{algorithm}[!h]
\caption{{\bf \harrangementc}\\
         {\bf Input:} arrangement tree $T$,  hyperplane $h$,  cell $c$,  constraints path to root $\sigma$}
\begin{algorithmic}[1]
\label{alg:harrangec}
    \IF{$T$ is null}
        \STATE $T=$ new {\bf ArrangementTree}($h$)
        \STATE $\sigma_l = \sigma\cup \{\sum_{k=1}^{d-1} h[k] \theta_k \leq 1 \}$
        \STATE $p=$ a point in $c$ s.t. $\sigma_l$ is satisfied
        \STATE {\bf if} $\mathcal{O}($OrderBy$_{p}(\mathcal{D}))=$ True {\bf then} {\bf return} $p$
        \STATE $\sigma_r = \sigma\cup \{\sum_{k=1}^{d-1} h[k] \theta_k \geq 1 \}$
        \STATE $p=$ a point in $c$ s.t. $\sigma_r$ is satisfied
        \STATE {\bf if} $\mathcal{O}($OrderBy$_{p}(\mathcal{D}))=$ True {\bf then} {\bf return} $p$
        \STATE {\bf return}
    \ENDIF
    \STATE $\sigma_l = \sigma\cup \{\sum_{k=1}^{d-1} T.h[k] \theta_k \leq 1 \}$
    \IF {$h$ passes through $\sigma_l$}
        \STATE {\bf if} $p=$\harrangementc($T$,$h$,$c$,$\sigma_l$) is not null {\bf then} {\bf return} $p$
    \ENDIF
    \STATE $\sigma_r = \sigma\cup \{\sum_{k=1}^{d-1} T.h[k] \theta_k \leq 1 \}$
    \IF {$h$ passes through $\sigma_r$}
        \STATE {\bf if} $p=$\harrangementc($T$,$h$,$c$,$\sigma_r$) is not null {\bf then} {\bf return} $p$
    \ENDIF
\end{algorithmic}
\end{algorithm}

After identifying $\mathcal{HC}$ (the sets of hyperplanes passing through the cells), for each cell $c\in Cells$, we limit the arrangement to $\mathcal{HC}[c]$.
Moreover, note that in this step our goal is to find a satisfactory function inside $c$. This is different from our objective in \arrangement, where we wanted to find {\em all} satisfactory regions.
This gives us the opportunity to apply a {\em stop early} strategy, as follows: at every iteration, while using the arrangement tree for the construction, check a function inside the newly added regions, and stop as soon as a satisfactory function is discovered.

Algorithm~\ref{alg:cellarrangement}, \cellarrangementnsp, assigns a satisfactory function to the cells that intersect with a satisfactory region $R$. It calls Algorithm~\ref{alg:harrangec} that adds the new hyperplanes and checks if a function inside the new regions is satisfactory. Both algorithms stop as soon as they find a satisfactory function and assign it to the cell.

Figure~\ref{fig:atc} illustrates how Algorithm~\ref{alg:cellarrangement} finds a satisfactory function for cell $c$. After adding hyperplanes $hc_1$ and $hc_2$, since functions $f_1$ to $f_6$ are unsatisfactory (denoted by red color), Algorithm~\ref{alg:cellarrangement} adds $hc_3$ to the construction. In this example, $hc_3$ does not pass through $\{hc_1^-, hc_2^-\}$, but it passes through $R=\{hc_1^-, hc_2^+\}$, dividing it into $R_l = R\cup hc_3^-$ and $R_r = R\cup hc_3^+$. Although $f_7\in R_l$ is unsatisfactory, $f_8\in R_r$ is satisfactory. The algorithm assigns $f_8$ to $c$ and stops without constructing the rest of the arrangement.

Considering $|\mathcal{HC}[c]|$ as the total number of hyperplanes passing through a cell $c$, the complexity their arrangement is $O(|\mathcal{HC}[c]|^{d-1})$. Thus, adopting Theorem~\ref{th:3} for Algorithm~\ref{alg:cellarrangement}, its time complexity is $O\big(|\mathcal{HC}[c]|^d Lp(|\mathcal{HC}[c]|) + |\mathcal{HC}[c]|^{d-1} n\log n \mathbb{O}_n \big)$ for a fixed $d$.
%, Algorithm~\ref{alg:cellarrangement} in the worst case needs to construct the complete arrangement, and is in $O\big(|\mathcal{HC}[c]|^{d-1} \big)$ for a fixed $d$.

\subsection{Coloring cells outside satisfactory regions}\label{subsec:coloring}
So far, we identified cells $\mathcal{C}$ that intersect with some satisfactory region, and assigned a satisfactory function to each of them. We now focus on cells $\bar{\mathcal{C}}$ that do not contain a satisfactory function. For  ease of explanation, we will represent the satisfactory function assigned to cell $c\in\mathcal{C}$ with the color of $c$ (see Figure~\ref{fig:coloring1}).
For each cell $c'\in \bar{\mathcal{C}}$, our objective is to find the closest satisfactory function to the center of $c'$, and to color $c'$ accordingly (see Figure~\ref{fig:coloring2}).

To do so, we implement \cellcoloring \submit{(see technical report~\cite{techreport} for pseudocode)},
%Algorithm~\ref{alg:colorcells}, 
an algorithm that uses monotonicity of the angular distance and adopts Dijkstra's algorithm~\cite{Dij}. The algorithm initially sets the distance of the satisfactory cells to zero, and the distance of all other cells to $\infty$, and adds them to a priority queue $Q$.
Then, while $Q$ is not empty, it visits the cell $c$ with the minimum distance, and remove it from $Q$.  For all neighbors of $c$ that are still not visited and their distances are more than the angular distance of their center with $F[c]$, the algorithm updates their distance and position in the queue, and sets their color to $F[c]$.

\techrep{
\vspace{2mm}
\begin{algorithm}[!h]
\caption{{\bf \cellcoloring}\\
         {\bf Input:} Satisfactory cells $\mathcal{C}$, unsatisfactory cells $\bar{\mathcal{C}}$, and assigned functions to cells $F$
        }
\begin{algorithmic}[1]
\label{alg:colorcells}
\FOR{$c \in Cells$}
	\STATE visited[$c$] = False
    \STATE {\bf if} $c\in\mathcal{C}$ {\bf then} $Q$.add\_with\_priority($c$, 0)
	\STATE {\bf else} $Q$.add\_with\_priority($c'$, $\infty$)
\ENDFOR
\WHILE {$Q$ is not empty}
	\STATE $c = Q$.extract\_min()
	\STATE visited$[c]$ = True
	\FOR{each neighbor $c'$ of $c$ where visited$[c]$ = False}
		\STATE alt = $\theta_{F[c],\mbox{center}(c')}$
		\IF{alt$<$dist$[c']$}
			\STATE $dist$[c'] = alt; $F[c'] = F[c]$
			\STATE $Q$.decrease\_priority$(c', alt)$
		\ENDIF
	\ENDFOR
\ENDWHILE
\end{algorithmic}
\end{algorithm}
}

Since the number of neighbors of each cell is fixed, it is easy to see that \cellcoloring is in $O(N\log N)$~\cite{Dij}.

Applying \cellcoloring completes offline preprocessing. After this step every cell in the partitioned angle space is assigned a satisfactory function\footnote{\small We assume the existence of at least one satisfactory region.}. % in function space.}.
We store the cell coordinates, together with the assigned satisfactory functions, as an index that enables online answering of user queries, discussed next.

\subsection{Online processing}\label{subsect:speedup-online}

Given an unsatisfactory function $f$, 
we now need to find the cell to which $f$ belongs, and to return the satisfactory function assigned to that cell. This is implemented by \techrep{Algorithm~\ref{alg:mdonline}} \submit{the algorithm \mdonline (see technical report~\cite{techreport} for pseudocode)}.

Given a query $f$ and the assigned functions to the cells, the algorithm transforms the weight vector of $f$ to polar coordinates and then performs binary search on each dimension to identify the cell $c$ to which $f$ belongs.  \mdonline returns the satisfactory function of the cell, $F[c]$.

\begin{theorem}\label{th:mdonline}
Algorithm \mdonline runs in $O(\log N)$ time.
\end{theorem}

\techrep{
\begin{algorithm}[!h]
\caption{{\bf \mdonline}\\
         {\bf Input:} partitioned space $T$, assigned functions $F$,  dataset $\mathcal{D}$, fairness oracle $\mathcal{O}$, and weight vector $\vec{w}$ \\
         {\bf Output:}  satisfactory weight vector $\vec{w'}$
        }
\begin{algorithmic}[1]
\label{alg:mdonline}
    \IF{$\mathcal{O}($OrderBy$_{f_{\vec{w}}}(\mathcal{D}))=$ True}
        \STATE {\bf return $\vec{w}$}
    \ENDIF
    \STATE $(r,\Theta)=$ {\bf ToPolar}$(\vec{w})$
    \FOR{$k=1$ to $d-1$}
        \STATE $T$ = apply binary search on children of $T$ and find the child to which $\Theta_k$ belongs
    \ENDFOR
    \STATE {\bf return} $F[T]$
\end{algorithmic}
\end{algorithm}
}

\techrep{
\begin{proof}
The proof simply follows the fact that ordering the items based on the input function is in $O(n\log n)$ while finding its corresponding cell, using binary search is in $O(\log N)$.
\end{proof}
}

\begin{theorem}\label{th:thresh}
Let $f_{opt}$ and $\theta_{opt}$ be the closest function and its angle distance to a queried function $f$.
Also, let $f_{app}$ and $\theta_{app}$ be the function and its angle distance that \submit{Algorithm \mdonline} \techrep{Algorithm~\ref{alg:mdonline}} returns for $f$, based on the space partitioning parameter $N$.
Then, $\theta_{app}\leq \theta_{opt}$ + $4 \arcsin \Bigg( \frac{\sqrt{d-1} }{2} \sqrt[d-1]{ \frac{\pi^{d/2}}{N 2^{d-1} \Gamma (d/2)} } ~ \Bigg)$.
\end{theorem}
\techrep{
\begin{proof}
\begin{figure}[h]
	\centering
    \includegraphics[width = 0.3\textwidth]{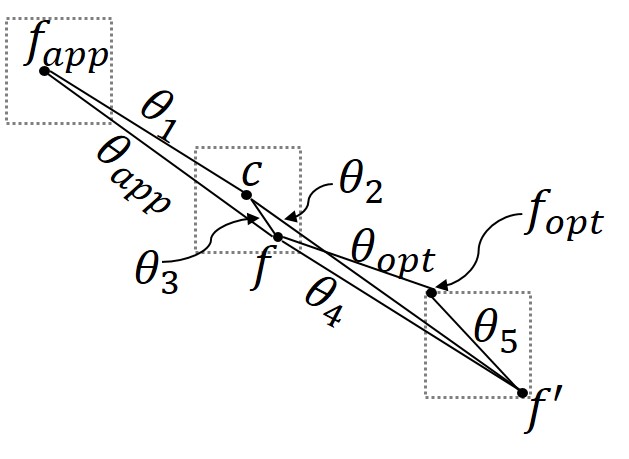}
    \vspace{-4mm}
    \caption{Illustration of $\theta_{app}$ v.s. $\theta_{opt}$}\label{fig:proofdist}
\end{figure}
Let $c_{app}$ and $c_{opt}$ be the cells $f_{app}$ and $f_{opt}$ belong to. 
First, there should exists a satisfactory function $f'$ inside $c_{opt}$ that is assigned to it. That is because $f_{opt}$ belongs to $c_{opt}$ and thus its intersection with the satisfactory regions is not empty.
Figure~\ref{fig:proofdist} illustrates such a setting. The point $c$ in the figure shows the center of the cell that $f$ belongs to. Since $f_{app}$ is assigned to this cell, the angle distance between $f_{app}$ and $c$ ($\theta_1$ in the figure) is less than the angle distance between $f'$ and $c$ ($\theta_2$ in the figure).
Let $\theta_3$, $\theta_4$, and $\theta_5$ (as specified in the figure) be the angle distance between $c$ and $f$, $f$ and $f'$, and $f'$ and $f_{opt}$, respectively.
Following the triangular inequality:
\begin{align}
\nonumber & \theta_{app} \leq \theta_1 + \theta_3,~ \theta_4 \geq \theta_2 - \theta_3 \\
\nonumber \Rightarrow ~ & \theta_{app} + \theta_2 - \theta_3 \leq \theta_1 + \theta_3 + \theta_4\\
\nonumber \Rightarrow ~ & \theta_{app} \leq \theta_4 + 2\theta_3
\end{align}
Similarly:
\begin{align}
\nonumber & \theta_4 \leq \theta_5 + \theta_{opt}\\
\nonumber \Rightarrow ~ & \theta_{app} \leq \theta_{opt} + \theta_5 + 2\theta_3
\end{align}
Let $\theta_r$ be the diameter of each cell. Looking at the figure, $\theta_5 \leq \theta_r$ and $\theta_3 \leq \theta_r /2$.
Thus: $$\theta_{app} \leq \theta_{opt} + 2 \theta_r$$
Following Equation~\ref{eq:cellside}, the diameter of the hypercube base of each cell is:
$$\eta_{d} = \sqrt{d-1} ~ \sqrt[d-1]{ \frac{\pi^{d/2}}{N 2^{d-1} \Gamma (d/2)} }$$
Therefore, $\theta_r$ is:
$$\theta_r =2 \arcsin \Bigg( \frac{\sqrt{d-1} }{2} \sqrt[d-1]{ \frac{\pi^{d/2}}{N 2^{d-1} \Gamma (d/2)} } ~ \Bigg)$$
Hence:
$$\theta_{app}\leq \theta_{opt} + 4 \arcsin \Bigg( \frac{\sqrt{d-1} }{2} \sqrt[d-1]{ \frac{\pi^{d/2}}{N 2^{d-1} \Gamma (d/2)} } ~ \Bigg)$$
\end{proof}
}

\submit{Proofs of Theorems~\ref{th:mdonline} and~\ref{th:thresh} are given in the technical report~\cite{techreport}.}

\subsection{Sampling for large-scale settings}
\label{sec:opt:sample}
A critical requirement of our system is to be efficient during online query processing, and it is fine for it to spend more time in the offline preprocessing.
As discussed in \S~\ref{sec:md} and \S~\ref{sec:speedup}, the proposed offline algorithms are polynomial for a fixed value of $d$.  In addition, the arrangement tree (c.f. \S~\ref{sec:md}) and the techniques of \S~\ref{sec:speedup} speed up preprocessing in practice.
However, preprocessing can still be slow, particularly for a large number of items. We reduce preprocessing time using sampling.

The main idea is that a uniform sample of the data maintains the underlying properties of the data distribution. Therefore, if a function is satisfactory for a dataset, it is {\em expected} to be satisfactory for a uniformly sampled subset.
Hence, for a datasets with large numbers of items, one can do the preprocessing on a uniformly sampled subset to find functions that are expected to be satisfactory for each cell.  We confirm the efficiency and effectiveness of this method experimentally on a dataset with over one million items in \S~\ref{sec:exp}. 
%As we shall later show in \S~\ref{sec:exp}, our experimental evaluation on a large dataset with more than one million items confirms this.

%\julia{What follows in this paragraph is defensive, I suggest removing, and replacing with a preview of the relevant experimental results instead.  I'll take care of this but let me know if you object.}\abol{I commented it, replacing with the few lines above}
%As a result, if one does the preprocessing on a uniformly sampled subset, rather than returning a function $f'$ that is satisfactory, when answering the users' queries, it returns a function $f'$ that is {\em expected} to be satisfactory. That is, the deviation of its ordering from a a satisfactory ordering, with a high confidence, is very low. Of course, providing the precise statements depends on the choice of ranking model.

%\input{discussions}
\section{Experimental Evaluation}\label{sec:exp}

\subsection{Experimental Setup}

\stitle{Hardware and platform.}
The experiments were performed on a Linux machine with a 2.6 GHz Core I7 CPU and 8GB memory.
The algorithms were implemented using Python2.7.
We used the Python scipy.optimize~\footnote{\small https://docs.scipy.org/doc/scipy/reference/optimize.html} package for LP optimizations.

\stitle{Datasets.} All experiments are conducted on real datasets.

%\noindent
{\it COMPAS}: a dataset collected and published by ProPublica as part of their investigation into racial bias in criminal risk assessment software~\cite{propublica}.
The dataset contains demographics, recidivism scores produced by the COMPAS software, and criminal offense information for 6,889 individuals.  

We used  {\tt c\_days\_from\_compas}, {\tt juv\_other\_count}, {\tt days\\\_b\_screening\_arrest}, {\tt start}, {\tt end}, {\tt age}, and {\tt priors\_cou-\\nt} as scoring attributes.  We normalized attribute values as $(val-min)/(max-min)$.  For all attributes except {\tt age}, a higher value corresponded to a higher score.
%While the other attributes are considered as the higher the better (i.e., a higher value dominates a lower one), we consider the {\tt age} as the lower the better (i.e., the lower the age is the higher the score will be).
In addition to the scoring attributes, we consider attributes {\tt sex} (0:male, 1: female), {\tt age\_binary} (0: less than 35 yo, 1: more than 36 yo), {\tt race} (0: African American, 1: Caucasian, 2: Other), and {\tt age\_bucketized} (0: less than 30 yo, 1: 31 to 40 yo, 2: more than 40 yo), as the type attributes.
%Besides its importance as a well-known dataset for fairness studies, this dataset has several scoring, large number of items, and protected attributes that allows us to study the performance of proposed algorithms for different scales and different fairness models on different combinations of attributes. Therefore,
COMPAS is the default dataset for our experiments.

%\noindent
{\it US Department of Transportation (DOT):} the flight on-time data-base published by DOT is 
widely used by third-party websites to identify the on-time performance of flights, routes, airports, and airlines~\cite{dot}.  The dataset contains 1,322,024 records, for all flights conducted by the 14 US carriers in the first three months of 2016. We use this dataset to study sampling for large-scale settings, and to showcase the application of our techniques for diversity.
%four major companies Delta Airlines (DL), American Airlines (AA), Southwest (WN), and United Airlines (UA)

\stitle{Fairness models.} We evaluate performance of our methods over two general fairness models, see  \S~\ref{sec:pre}. 

%\noindent
{\it FM1}, proportional representation on a single type attribute, is the default fairness model in our experiments.  This model can express common proportionality constraints from the literature~\cite{DBLP:conf/innovations/DworkHPRZ12,DBLP:conf/kdd/FeldmanFMSV15,DBLP:journals/datamine/Zliobaite17}, including also for ranked outputs~\cite{DBLP:conf/cikm/ZehlikeB0HMB17} and for set selection~\cite{StoyanovichYJ18}.  The distinguishing features of FM1 are (1) that the type attribute partitions the input dataset $\mathcal{D}$ into groups and (2) that the proportion of members of a particular group is bounded from below, from above, or both.  For the COMPAS dataset, unless noted otherwise, we state FM1 over the type attribute {\tt race} as follows:  African Americans constitute about 50\% of the dataset; a fairness oracle will consider a ranking to be satisfactory if at most 60\% (or about 10\% more than in $\mathcal{D}$) of the top-ranked 30\% are African American. 

%\noindent
{\it FM2}, proportional representation on multiple, possibly overlapping, type attributes, is a generalization of FM1 that can express proportionality constraints of~\cite{DBLP:journals/corr/CelisSV17}.  As in~\cite{DBLP:journals/corr/CelisSV17}, we bound the number of members of a group from above.  For example, for COMPAS, we specify the maximum number of items among the top-ranked 30\% based on {\tt sex} (80\% of $\mathcal{D}$ are male), {\tt race} (50\% are African American), and {\tt age\_bucketized} (42\% are 30 years old or younger, 34\% are between 31 and 50, and 24\% are over 50).  In all experiments, a ranking is considered satisfactory if the proportion of members of a particular demographic group is no more than 10\% higher than its proportion in $\mathcal{D}$.

\begin{figure*}[!ht] 
	%\centering
    \begin{minipage}[t]{0.24\linewidth}
        \centering
        \includegraphics[width =\textwidth]{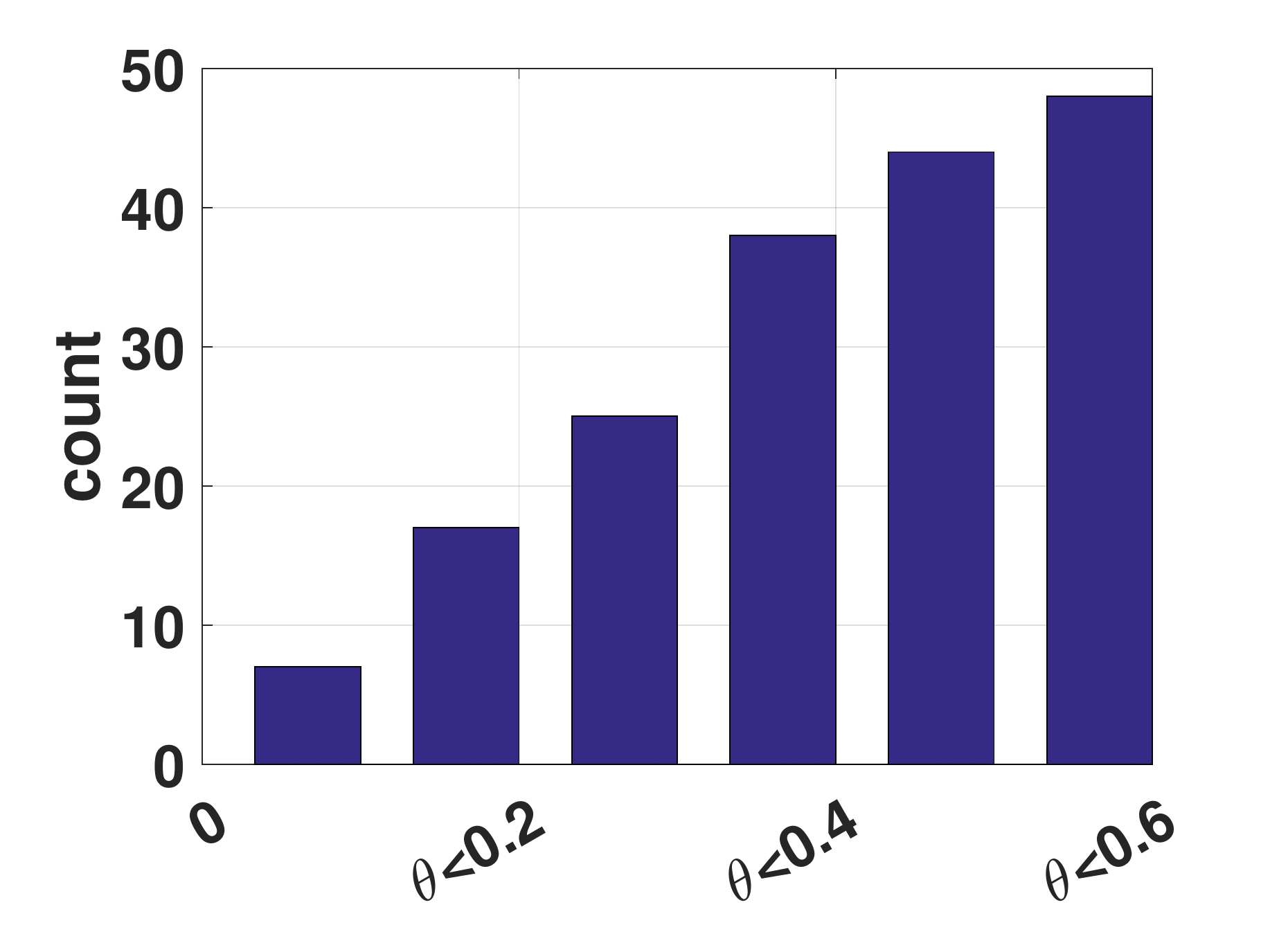}
        \vspace{-6mm}\caption{MD, angle distance between input and output functions}
        \label{fig:exp-mdonlinedist}
    \end{minipage}
    \hspace{1mm}
    \begin{minipage}[t]{0.24\linewidth}
        \centering
        \includegraphics[width =\textwidth]{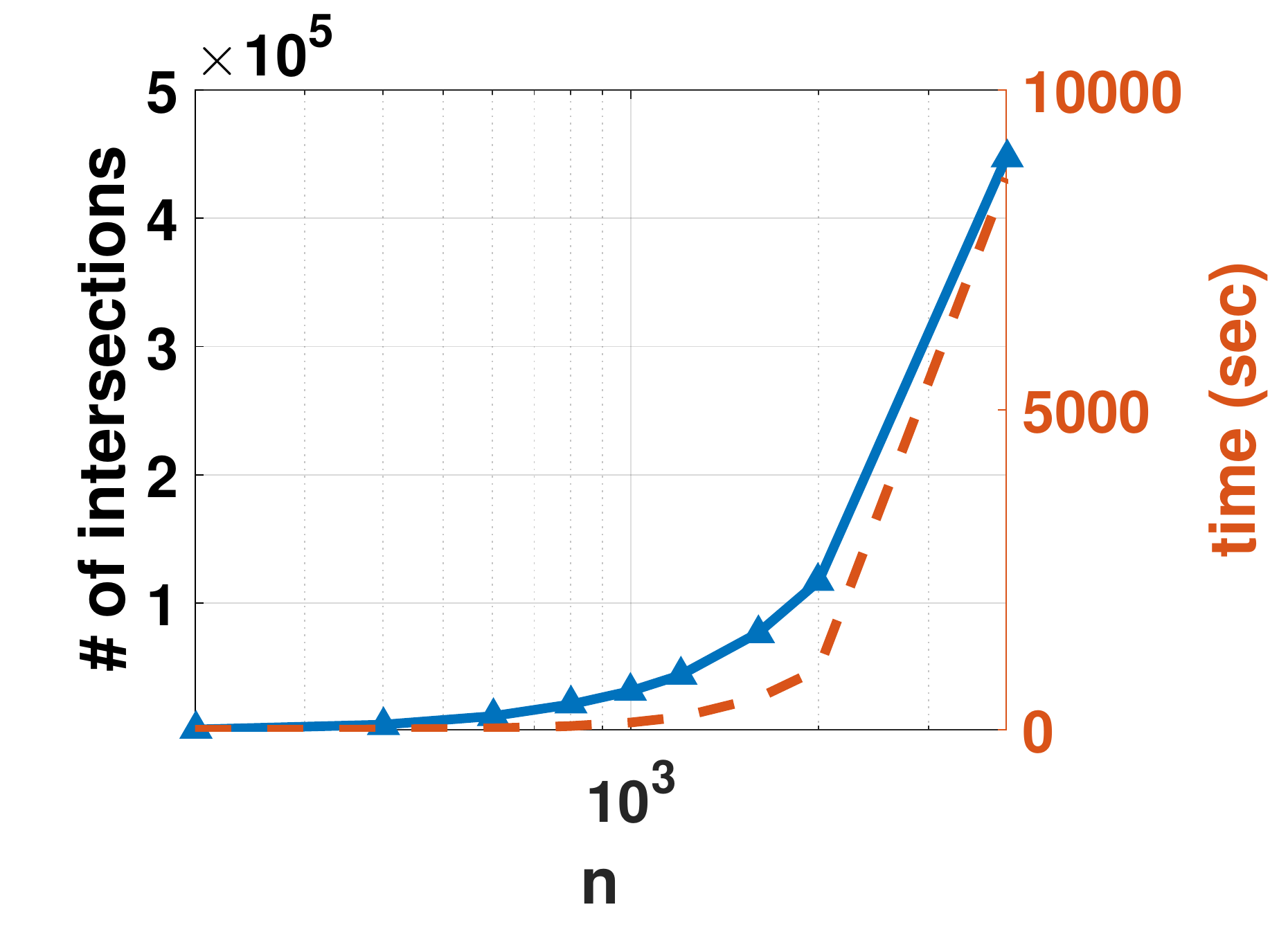}
        \vspace{-6mm}\caption{2D; preprocessing time, varying n}
        \label{fig:exp-twod1}
    \end{minipage}
    \hspace{1mm}
    \begin{minipage}[t]{0.24\linewidth}
        \centering
        \includegraphics[width =\textwidth]{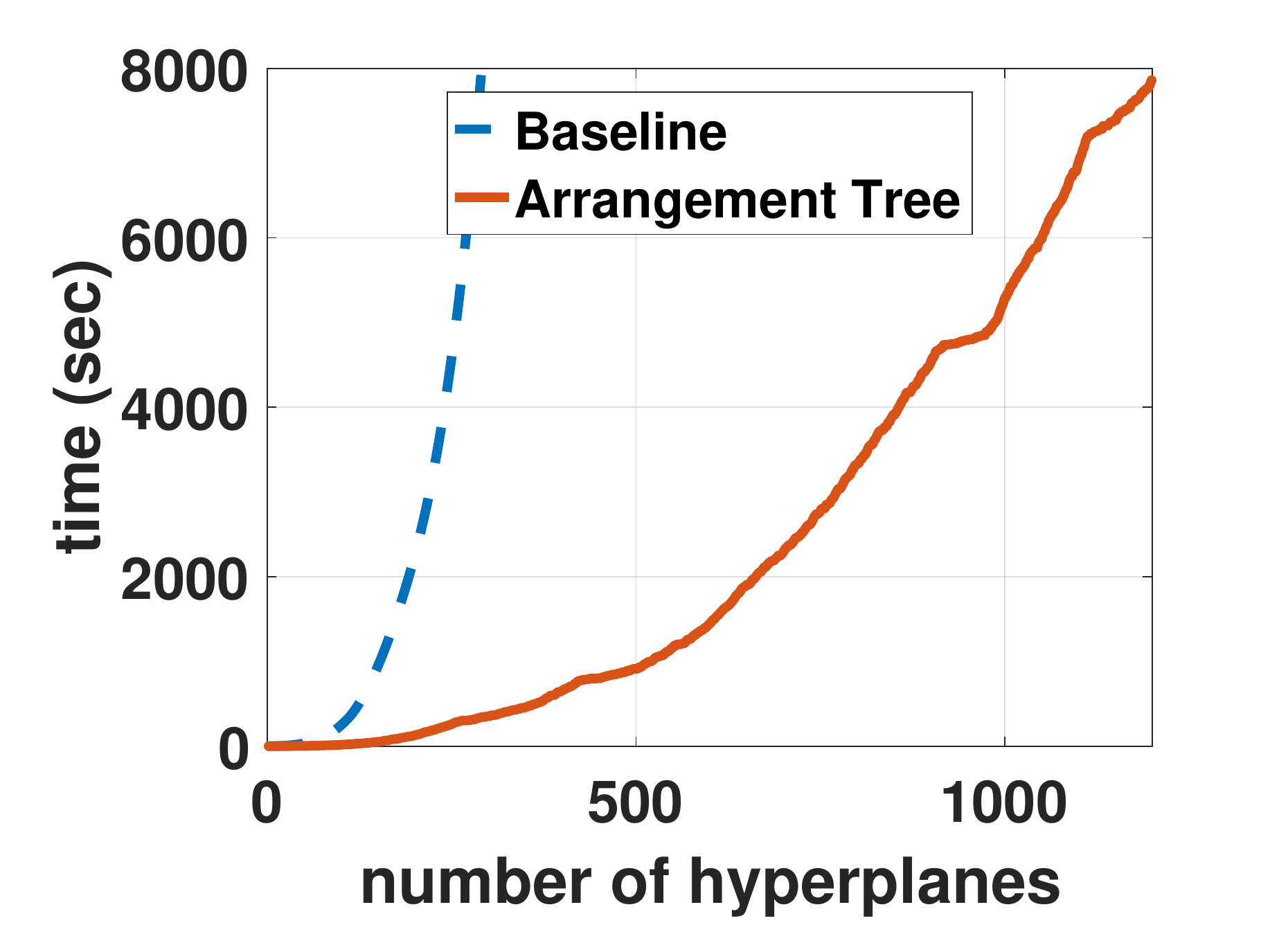}
        \vspace{-6mm}\caption{MD; arrangement construction cost, the advantage of using arrangement tree}
        \label{fig:exp-arrangementtree}
    \end{minipage}  
    \hspace{1mm}
    \begin{minipage}[t]{0.24\linewidth}
        \centering
        \includegraphics[width =\textwidth]{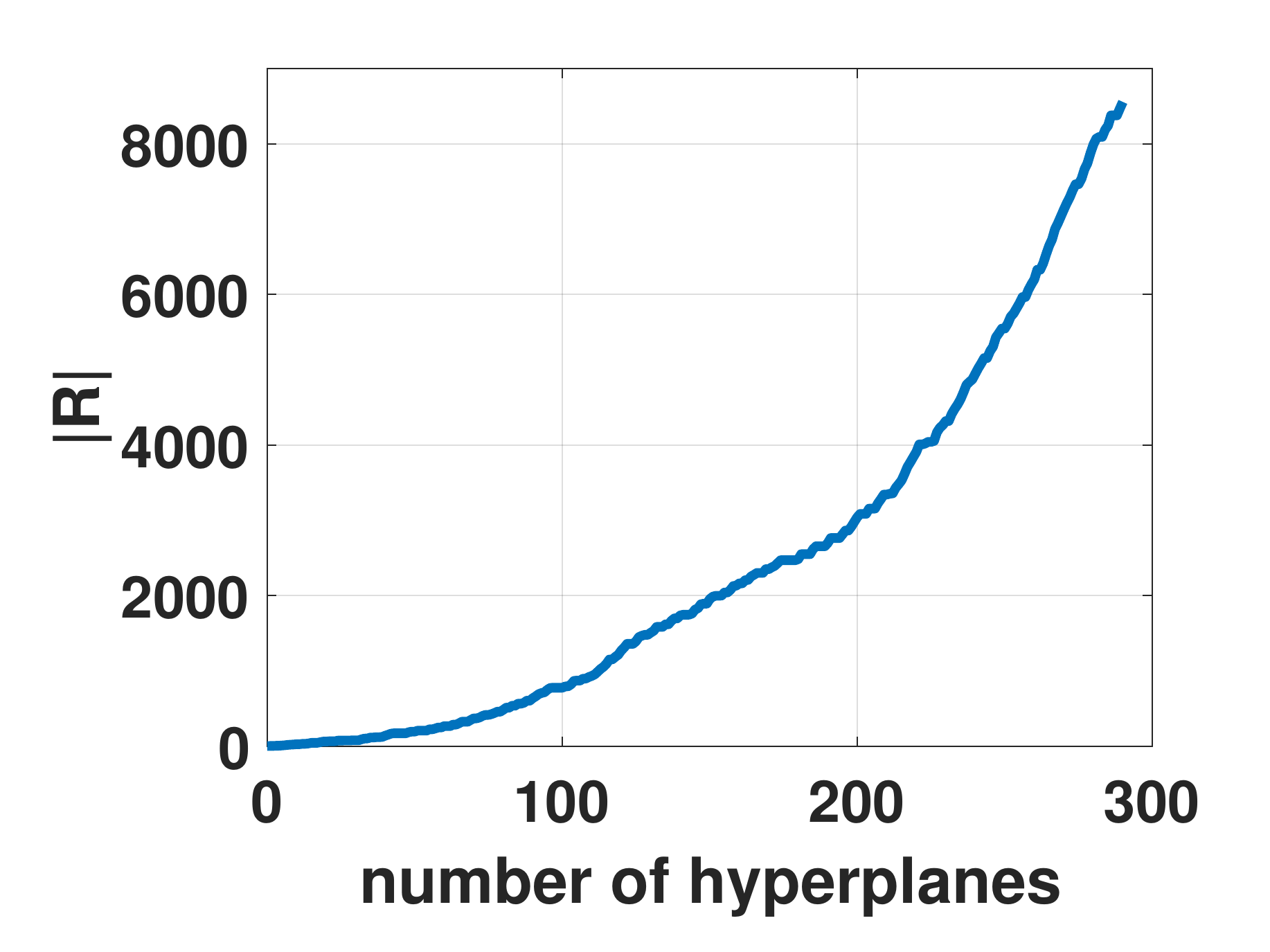}
        \vspace{-6mm}\caption{MD; arrangement complexity while adding the hyperplanes ($d=3$)}
        \label{fig:exp-Noregions}
    \end{minipage}
\end{figure*}

\begin{figure*}[!ht] 
    \begin{minipage}[t]{0.24\linewidth}
        \centering
        \includegraphics[width =\textwidth]{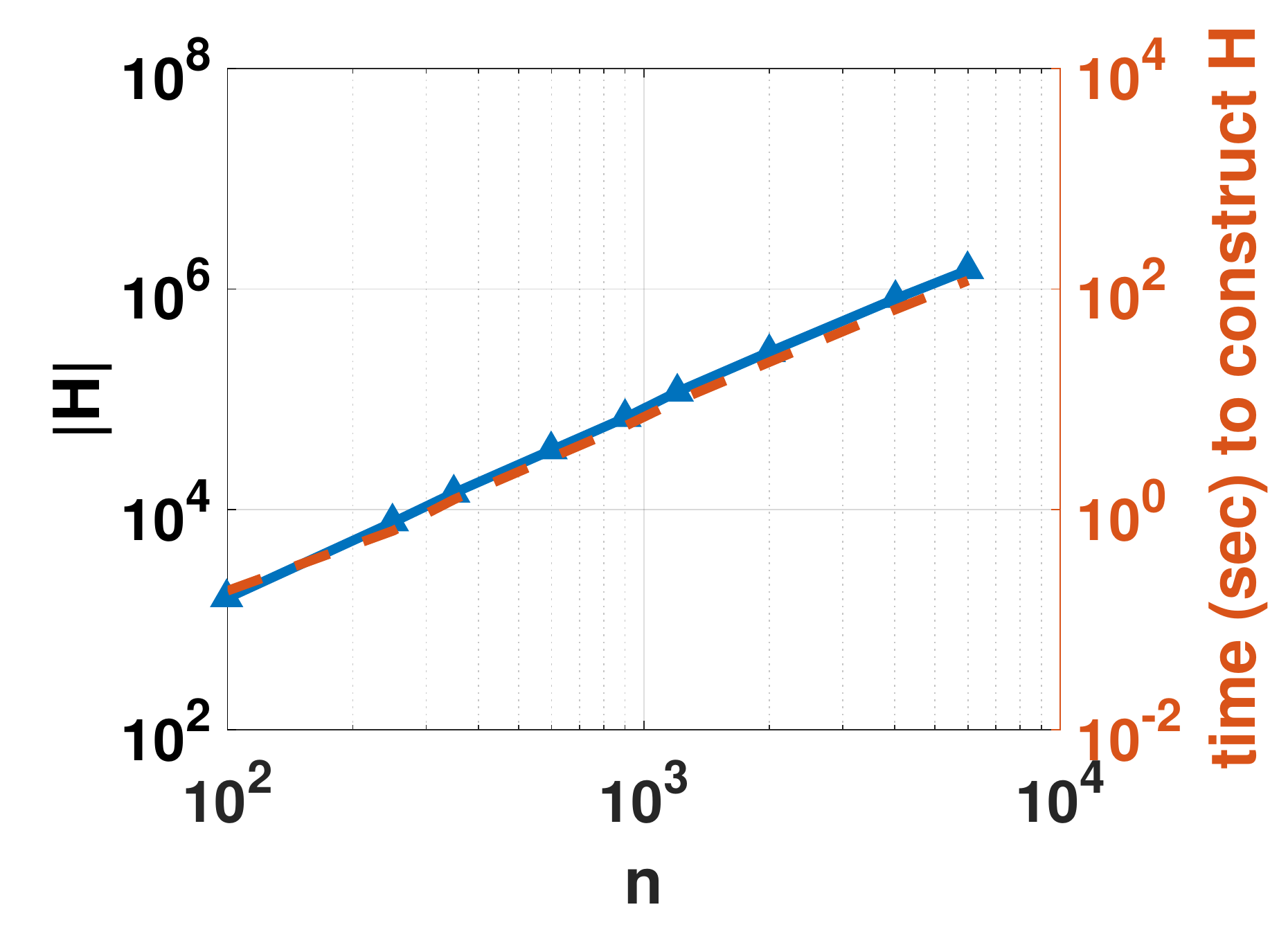}
        \vspace{-6mm}\caption{MD; effect of $n$ on $|H|$ ($d=3$)}
        \label{fig:exp-mdh}
    \end{minipage}
    \hspace{1mm}    
    \begin{minipage}[t]{0.24\linewidth}
        \centering
        \includegraphics[width =\textwidth]{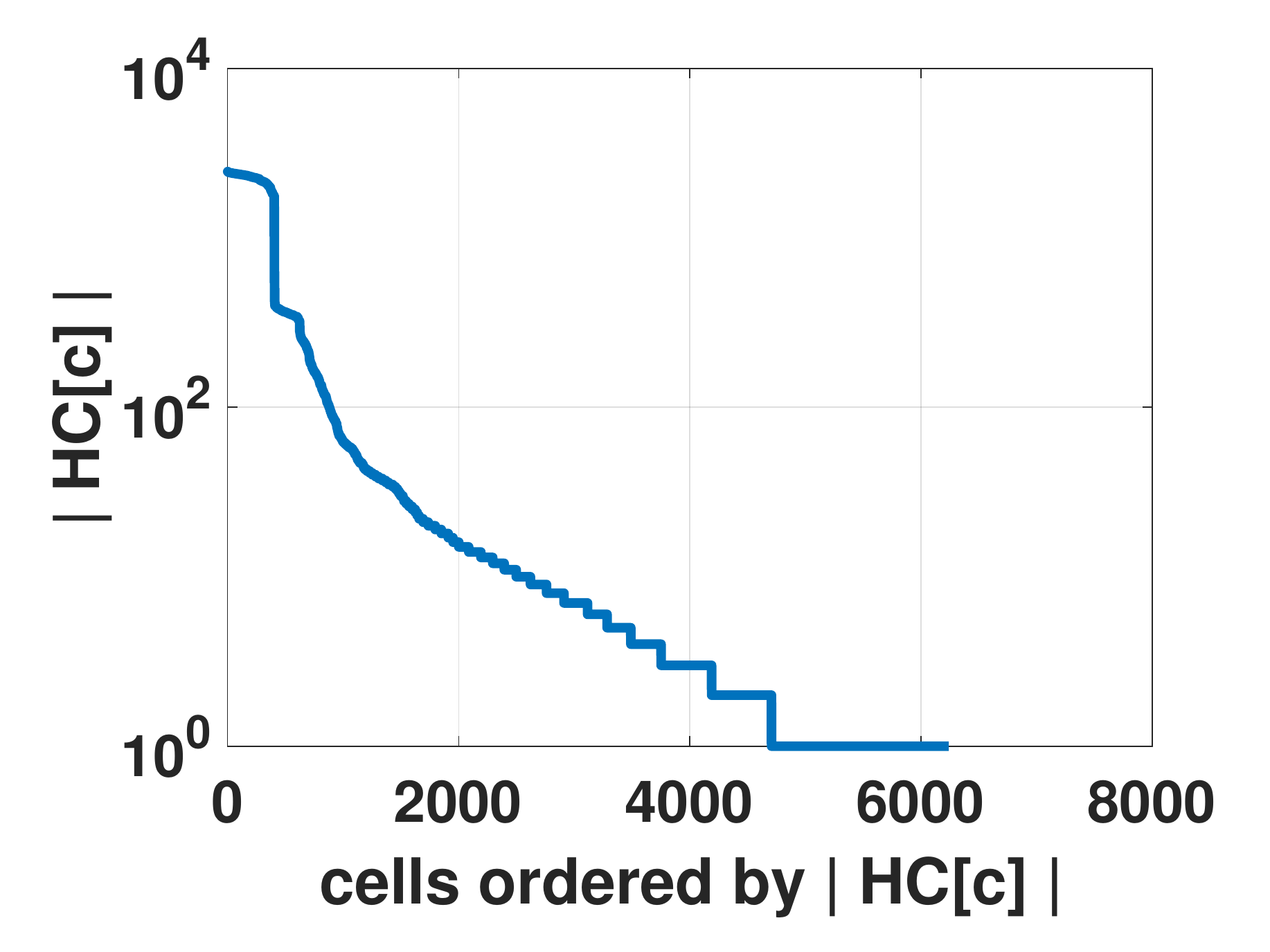}
        \vspace{-6mm}\caption{MD; umber of hyperplanes passing through each cell ($n=100$, $d=4$)}
        \label{fig:exp-mdc}
    \end{minipage}
    \hspace{1mm} 
    \begin{minipage}[t]{0.24\linewidth}
        \centering
        \includegraphics[width =\textwidth]{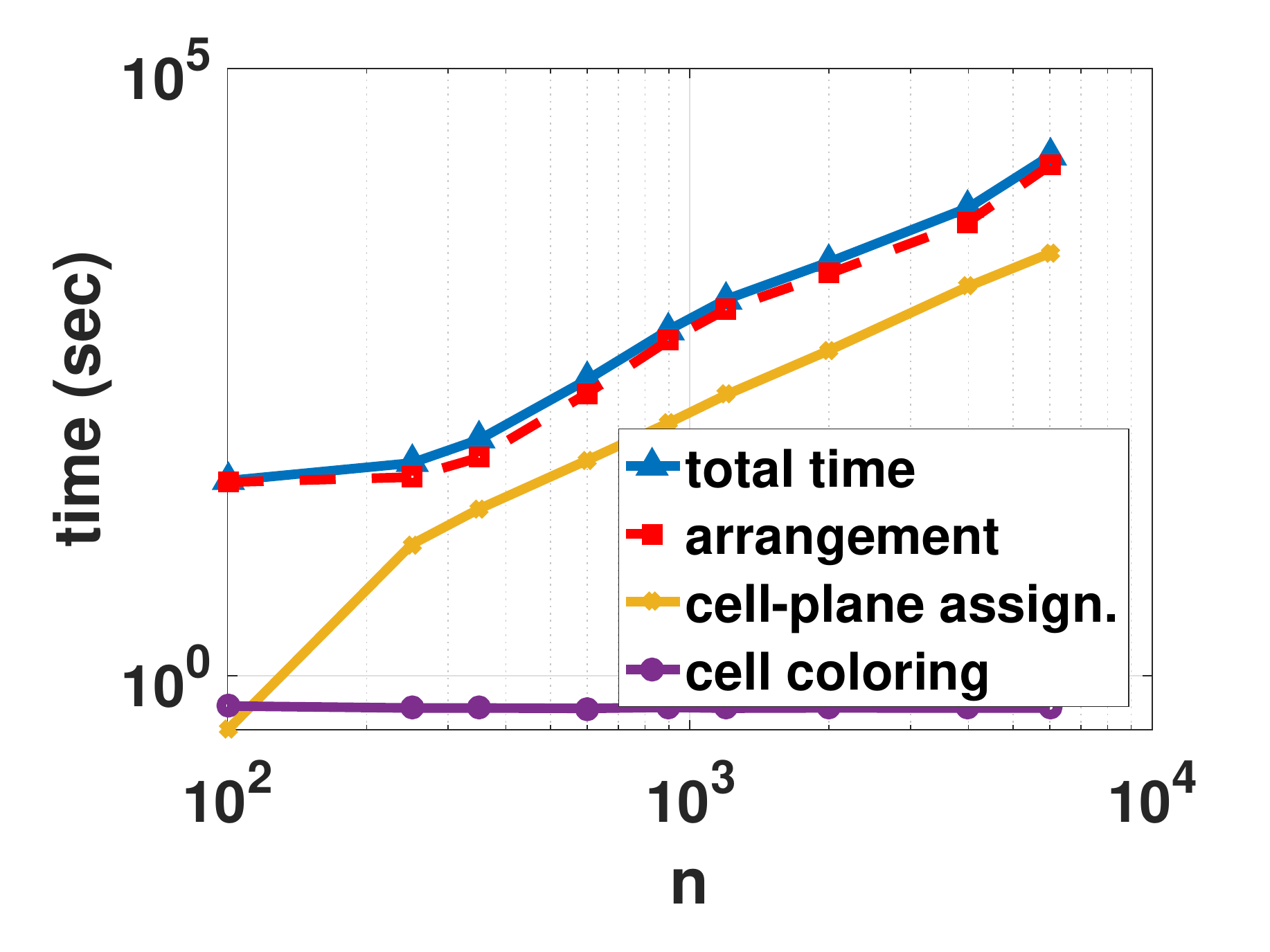}
        \vspace{-6mm}\caption{MD; effect of $n$ on preprocessing times for different steps; $d=3$}
        \label{fig:exp-mdcost}
    \end{minipage}
    \hspace{1mm}
    \begin{minipage}[t]{0.24\linewidth}
        \centering
        \includegraphics[width =\textwidth]{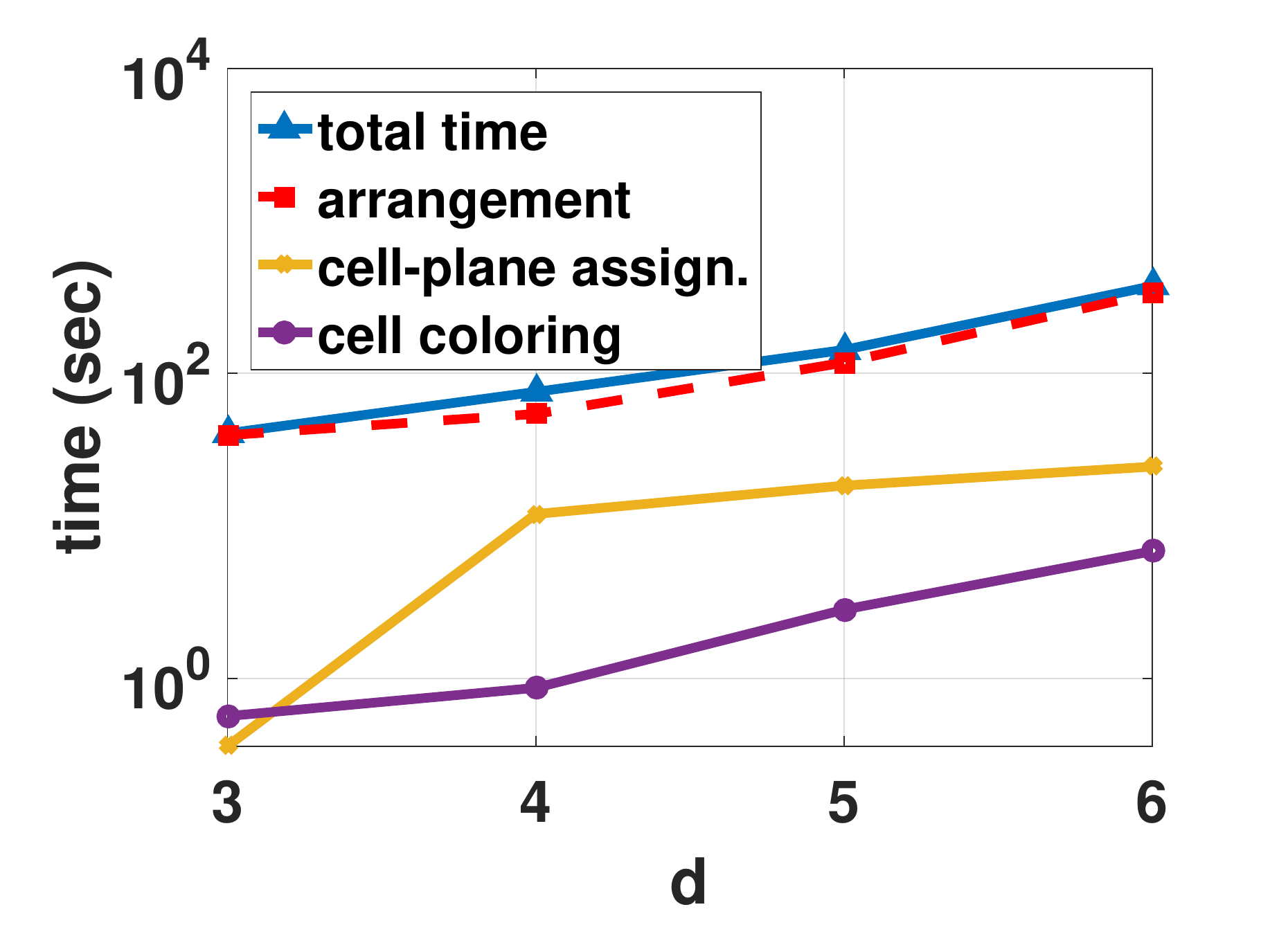}
        \vspace{-6mm}\caption{MD; effect of $d$ on preprocessing times for different steps; $n=100$}
        \label{fig:exp-mdmcost}
    \end{minipage}
\end{figure*}

%\subsection{Experimental Results}

\subsection{Validation experiments}\label{subsec:exp-valid}
%\julia{A systematic validation experiment should go first.  Need a figure here: a cumulative bar chart: $dist(f, f')$ on the $x$-axis (normalized to $[0,1]$), and number of cases where distance was at most $x$ on the $y$-axis.}
In our first experiment, we show that our methods are effective --- that they can identify scoring functions that are both satisfactory and similar to the user's query.  We use the COMPAS dataset with $d=3$ (scoring attributes  {\tt start}, {\tt c\_days\_from\_compas}, {\tt juv\_other\_count}, {\tt start}), and with fairness model FM1 on {\tt race} (at most 60\% African Americans among the top 30\%). 

%To do the systematic validation, using COMPAS while setting $d=3$, we issued 100 random queries and passed them to our system. We consider an ordering to be fair, if it contains at most 60\% in its top-30\%.
% 52\% of the input functions were already satisfactory. 

 We issued 100 random queries, and observed that 52 of them were satisfactory, and so no further intervention was needed.  For the remaining 48 functions, we used our methods to suggest the nearest satisfactory function.  Figure~\ref{fig:exp-mdonlinedist} presents a cumulative plot of the results for these 48 cases, showing the angle distance  $\theta(f,f')$ between the input $f$ and the output $f'$ on the $x$-axis, and the number of queries with {\em at most} that distance on the $y$-axis.  
 
We observe that a satisfactory function $f'$ was found close to the input function $f$ in all cases.
Specifically, note that $\theta(f,f') < 0.6$ in all cases, and recall that $\theta \in [0, \pi/2]$, with lower values corresponding to higher similarity.  (For a more intuitive measure: the value of $\theta = 0.6$ corresponds to cosine similarity of 0.82, where 1 is best, and 0 is worst).  Among the 48 cases, 38 had $\theta(f,f') < 0.4$ (cosine similarity 0.92).

%we plotted the cumulative angle distance between the input and output functions in Figure~\ref{fig:exp-mdonlinedist}. As shown in the figure, the angle distance between the input and output functions were less than 0.6 in all cases, that is the cosine similarity of the system output with what the queried function was always more than 0.82\%. This number was more than 0.92\% for 90\% of the queried functions.

In our next experiment, we give an intuitive understanding of the layout of satisfactory regions in the space of ranking functions.  We use COMPAS with {\tt age} (lower is better) and {\tt juv\_other\_count} (higher is better) for scoring.  The intuition behind this scoring function is that individuals who are younger, and who have a higher number of juvenile offenses, are considered to be more likely to re-offend, and so may be given higher priority for particular supportive services or interventions.

Naturally, a scoring function that associates a high weight with {\tt age} will include mostly members of the younger age group at top ranks.   About 60\% of COMPAS are 35 years old or younger. Consider a fairness oracle that uses FM1 over {\tt age\_binary} (with groups $g_1$: 35 year old or younger, and $g_2$: over 35 years old), and that considers a ranking satisfactory if at most 70\% of the top-100 results are in $g_1$.  Because of the correlation (by design) between one of the scoring attributes and the type attribute, there is only one satisfactory region for this problem set-up --- it corresponds to the set of functions in which the weight on {\tt age} is close to 0, and with the angle with the $x$-axis ({\tt juv\_other\_count}) of at most 0.31.

%Next, we study the effect correlation between a ranking attribute and the protected attribute. To do so, we choose the attributes {\tt priors\_count} and {\tt age} for scoring while choosing the attributes {\tt age\_binary} as the protected attribute. 4,065 out of 6,889 persons in the dataset belong to age group $0$. We considered an ordering to be fair, if at most 68 of its top 100 are less than 35 yo.
%It turned out that for this setting there is only one satisfactory region that is the set of function that their angle with the x-axis ({\tt priors\_count}) is at most 0.4. That is because of the high correlation between the attributes {\tt age} and {\tt age\_binary}, which makes only the functions that put a weight close to $0$ to the attribute {\tt age} to be fair.

Next, suppose that we use the same scoring attributes, but a different fairness oracle --- one that applies FM1 on the attribute {\tt race}, requiring that at most 60 of the top-100 are African American.  This time, there exist several satisfactory regions.  In fact, for any assignment of weights to the two scoring attributes, there exists a satisfactory function $f'$ such that $\theta(f, f') < 0.11$ (cosine similarity between $f$ and $f'$ is always more than 0.99).  

%In the next experiment, we replaced the protected attribute with {\tt race} and considered an ordering to be fair, if it contains at most 60 persons in its top 100.
%This time there exist several satisfactory ranges in the data while there is a small gap between them. 
%In this exp, the max. angle distance between a none-satisfactory function and our system's output is less than 0.11. It means that the cosine similarity of our suggestion and user input is always more than 0.99.
In our final validation experiment, we use {\tt juv\_other\_count} and {\tt c\_days\_from\_compas} for scoring, with fairness model FM2 that considers a ranking satisfactory if there are at most 90 males, at most 60 African Americans, and at most 52 persons who are 30 years old or younger at the top-100.  This fairness model is stricter than in the  preceding experiment (with FM1 on {\tt race}), making the gaps between the satisfactory regions wider. Still, the maximum angle between $f$ and $f'$ was less than 0.28, which corresponds to the minimum cosine similarity of 0.96.

\subsection{Performance of query answering}
\label{sec:exp:online}
While preprocessing can take more time, a critical requirement of our system is to be fast when answering users' queries. In this section, we use the COMPAS dataset and evaluate the performance of \twodonline and \mdonline, the two-dimensional and multi-dimensional algorithms for online query answering.  We show that queries can be answered in interactive time.
We use the default fairness model (i.e., at most 60\% AA in the top-30\%) and the scoring attributes in the same ordering provided in the description of COMPAS dataset.
%\julia{Which scoring attributes and which fairness model were used in these experiments?  Can we say that results are representative of other combinations of scoring attributes?}

\stitle{2D.} One nice property of \twodonline is that it does not need to access the raw data at query time. It only needs to apply binary search on the sorted list of satisfactory ranges to locate the position of the input function $f$. 
In this experiment, we compare the required time for ordering the results based on the input function, averaged over 30 runs of \twodonline on random inputs. Confirming the theoretical $O(\log n)$ complexity of \twodonline v.s. the $O(n\log n)$ for the ordering, \twodonline only required {\em 30 $\mu$sec} on average, while even ordering the results based on $f$ (to check if $f$ is satisfactory) required 25 msec to complete.

\stitle{MD.}
In this experiment, similarly to 2D, we took 
%While not being very efficient in the preprocessing is tolerable, it is critical to be efficient when answering the users' queries.
%Figure~\ref{fig:exp-mdonline} shows 
the average running time of 30 random queries, for between 3 and 6 scoring attributes (dimensions).
Upon arrival of a query function $f$, \mdonline 
%first orders the items based on $f$, in time $O(n\log n)$, and checks if the input function itself is satisfactory. If not, \mdonline 
finds the cell to which $f$ belongs in $O(\log N)$, and returns the corresponding satisfactory function $f'$.
As a result,
similar to \twodonline is significantly faster than even finding the ordering of the items based on $f$. This is confirmed  in our experiments were the running time, in all cases, was less than {\em 200 $\mu$sec} whereas the time required to order the items based on $f$ was {\em 25 msec}. Please note that the running time of \mdonline is independent of $n$ (the number of items in the dataset) and will perform similarly for the very large datasets.

%\begin{figure}[!ht] 
%        \centering
%        \includegraphics[width =0.35\textwidth]{plots/MD2}
%        \caption{MD; effect of $d$ on online query answering time}
%        \label{fig:exp-mdonline}
%\end{figure}

\subsection{Performance of preprocessing}\label{subsec:exp-performance}
In order to study the preprocessing performance, similar to \S~\ref{sec:exp:online}, we use COMPAS as the default dataset, the default fairness model (at most 60\% African Americans at the top-ranking 30\%), and the scoring attributes in the same ordering provided in the description of COMPAS dataset.

\stitle{2D.}
We start by evaluating the efficiency of \twodraynsp, the 2D preprocessing algorithm proposed in \S~\ref{sec:2d}.
%As our default dataset, we chose COMPAS with scoring attributes {\tt priors\_count} and {\tt age}, under the fairness model FM1 on {\tt race} (at most 60\% African Americans at the top-ranking 30\%). 
We study the effect of $n$ (the number of items in the dataset) on the performance of the algorithm \twodray and evaluate the number of \rankswitches and the running time of it.
Figure~\ref{fig:exp-twod1} shows the experiment results for varying the number of items from 100 to 6,000. 
The $x$-axis shows the values of $n$ (in log-scale), and the left and right $y$-axes show the number of \rankswitches and the running time of \twodraynsp, respectively. 
Looking at the left $y$-axis, one can observe that the number of \rankswitches is much smaller than the theoretical $O(n^2)$ upper-bound. For example, while the upper-bound on the number of \rankswitches for $n=4$k is $16$M, the observed number in this experiment was $450$k. This is because the pairs of items in which one dominates the other do not have a \rankswitchnospace.  Also, looking at the right $y$-axis, and comparing the dashed orange line (time) with the blue line (number of \rankswitchesnospace), one can see that the orange line has a sharper slope as it passes through the blue line. 
This is because 
the oracle is in $O(n)$ and thus, based on Theorem~\ref{th:2dray}, \twodray is in $O(n^3)$. %This confirms the analysis provided in .

\stitle{MD, the effect of using arrangement tree.}
In \S~\ref{sec:md}, we proposed the arrangement tree data structure
for constructing the arrangement of hyperplanes, in order to skip comparing a new hyperplane with all  current regions. 
Here, as the first MD experiment, we run the algorithm \arrangement as the baseline and also use \harrangement for adding the hyperplanes using the arrangement tree. %\julia{Which dataset, which scoring attributes, which fairness oracle is used here?}

Figure~\ref{fig:exp-arrangementtree}
shows the incremental cost of adding hyperplanes to the arrangement when $d=3$.
While the baseline  (\arrangementn) needed 8,000 seconds for adding the first 250 hyperplanes, using the arrangement tree helped save around 7,740 seconds.
Fixing the budget to 8,000 seconds, the baseline could construct the arrangement for the first 250 hyperplanes, while using the arrangement tree allowed us to extend the construction to 1,200 hyperplanes. 

Recall from \S\ref{sec:md} that the number of regions at step $i$ is $O(i^{2(d-1)})$, and hence, adding the consequent hyperplanes (with \arrangementn) is more expensive. This is presented in Figure~\ref{fig:exp-Noregions}, where the $y$-axis shows the number of regions in the arrangement ($|\mathcal{R}|$) for different number of hyperplanes. 
Observe that the number of regions for the first 50 hyperplanes is less than 200; it increases to  more than 5,000 regions for the hyperplanes that are added after 250$^{th}$ iteration. As a result, while adding a hyperplane (without using the arrangement tree) at the first 50 iterations requires checking fewer than 200 regions, adding a hyperplane after iteration 250 requires checking more than 5,000 regions, and so is significantly more expensive.

\stitle{MD, preprocessing.}
%After studying the cost of constructing arrangement of hyperplanes, here we 
We now evaluate the algorithms proposed in \S~\ref{sec:speedup} for preprocessing the data in  partitioned angle space.
%\julia{Is $|H|$ the number of hyperplanes?  Need to say this.}
First, similar to the 2D experiments, varying $n$ from 200 to 6,000, in Figure~\ref{fig:exp-mdh} we observe $|H|$ (the number of hyperplanes) as well as the time for constructing the hyperplanes in the angle coordinate system.
%\julia{I don't understand how to compare these two figures, they have different parameters on the $y$-axis.}
Comparing this figure with Figure~\ref{fig:exp-twod1} (remember that intersections in 2D and hyperplanes in MD refer to the \rankswitchesnospace), we observe that $|H|$ gets closer to $n^2$ as the number of dimensions increase.
This is because, as the number of dimensions increases, the probability that one in a pair of items dominate the other decreases, and therefore $|H|$ gets closer to $n^2$. Also, looking at the right-y-axis and the dashed orange line and comparing it with $|H|$ (the left-y-axis) confirms that the total running time is linear to the number of hyperplanes. %\julia{"linear constant" in the previous sentence; which is it?}

In the previous experiment for observing the benefit of using the arrangement tree, we discussed the effect of the number of hyperplanes on the complexity of the arrangement (quantified by the number of regions) and on the running time for constructing it. 
Thus, rather than constructing the arrangement for the complete set of hyperplanes, in \S~\ref{sec:speedup}, we limit the arrangement construction for each cell to the hyperplanes passing through it.  In Figure~\ref{fig:exp-mdc} we set the number of items to $100$ and $d$ to $4$, and
observe the number of hyperplanes passing through the cells. The $x$-axis in Figure~\ref{fig:exp-mdc} is the cells sorted by $|\mathcal{HC}[c]|$ (the number of hyperplanes passing through the cell $c$), and the $y$-axis shows $|\mathcal{HC}[c]|$ for each cell $c$.
Looking at the figure, one can see that more than 5000, out of 6000 cells have less than 100 hyperplanes passing through them, and even constructing the complete arrangement inside them is not very expensive.
We explained in \S~\ref{sec:speedup} that our goal is to associate a satisfactory function with each cell, allowing \cellarrangement to {\em stop early} (before constructing the complete arrangement) once a satisfactory function is identified. 
%, the objective while checking each cell is to assign a satisfactory function to it (if it intersects with a satisfactory region), not constructing the arrangement of hyperplanes passing through it. 
%Therefore, \cellarrangement tries to {\em stop early} before the complete construction of the arrangement. 
%\julia{What follows is an informal description of an algorithm that's already described elsewhere.  What point do we need to make here, succinctly?  Can we skip the rest of the paragraph, from here and until "Figures 21 and 22"?} 
%Therefore, for each region, it incrementally adds the hyperplanes intersecting it to the arrangement (using the arrangement tree); checks . As a result, \cellarrangement needs to construct the complete arrangement if the cells does not intersect with the satisfactory regions.

Figures~\ref{fig:exp-mdcost} and~\ref{fig:exp-mdmcost} show the required time for different steps of preprocessing, as well as the total preprocessing time. Figure~\ref{fig:exp-mdcost} shows the cost for varying $n$, with $d=3$ and $N=40,000$. In Figure~\ref{fig:exp-mdmcost} we fix $n=100$ and $N =40,000$ and vary $d$.  The yellow line in both figures shows the required time for identifying the hyperplanes passing through each cell. Applying \cellplane for finding the cells for each hyperplane helps skip a large portion of the cells. Still its running time increases significantly as $n$ increases. This is because the number of hyperplanes $|H|$ is in $O(n^2)$. %\julia{I say "the number of hyperplanes $|H|$, correct?  If so, let's remove $|H|$.}
On the other hand, despite the complexity of the arrangement construction (c.f. Theorem~\ref{th:3}), finding a satisfactory function for each cell that intersects with a satisfactory region (the dashed red line) may not be not very expensive and in certain cases has similar running time as \cellplanensp. Different optimizations proposed in \S~\ref{sec:md} and~\ref{sec:speedup} result in reasonable performance of this step. First, reducing the construction of the arrangement for each cell $c$, to the hyperplanes passing through it, reduces the complexity of the arrangement to $|\mathcal{HC}[c]|^{d-1}$. Second, as shown in Figure~\ref{fig:exp-arrangementtree}, the arrangement tree data structure helps
to rule out checking the intersection of the hyperplanes with all regions. %not to compare the new hyperplanes all the current regions. \julia{What does "helps not to compare the new hyperplanes all the current regions" mean?} 
Finally, the early stop condition is effective at reducing the running time. %in the arrangement construction for each cell (that is issuing a function for each new region and stopping the construction as soon as a satisfactory function is discovered) helps not to construct the complete arrangement. 
Still, looking at Figures~\ref{fig:exp-mdcost} and~\ref{fig:exp-mdmcost} this step always takes the majority of the preprocessing time.

The final step is to use \cellcoloring to assign the satisfactory function of the closest satisfactory cell to each unsatisfactory cell. Using a priority queue, this step is expected to be fast, which is observed in all the settings in Figures~\ref{fig:exp-mdcost} and~\ref{fig:exp-mdmcost}.
%Even though arrangement tree, limiting the arrangement construction to the cell scope, and early stop in the cells help in reducing the arrangement construction cost significantly, still the total preprocessing cost is dominated by the arrangement construction cost, as in both Figures~\ref{fig:exp-mdc} and~\ref{fig:exp-mdmcost}, the red and blue lines overlap. We again want to emphasize that efficiency during the preprocessing is not crucial for us and our aim is to be fast in online query processing.
%Still, as discussed in \S~\ref{subsec: sampling},
%for the situations where the input size is large that the preprocessing is not affordable, one can apply sampling, as it provides reasonable, ``close to satisfactory'' results.
%Next, we turn our attention to the system performance during the online query answering phase. \abol{add the results for }

\begin{comment}
- effect of number of cells in the arrangement: change the number of cells and monitor the cost: set n=50 for example and 3d
\end{comment}

\stitle{MD, sampling for a large-scale setting.}
%In \S~\ref{sec:discussions}, we discussed applying this proposal for diversity. 
We discussed in \S~\ref{sec:opt:sample} that preprocessing time can be reduced for very large datasets by conducting it over a uniform sample.  In this experiment, we use the DOT dataset, with three scoring attributes, {\tt departure delay}, {\tt arrival delay}, and {\tt taxi in}. The fairness oracle uses FM1 with airline name as the type attribute.  A ranking is satisfactory if the percentage of outcomes from each of four major companies Delta Airlines (DL), American Airlines (AA), Southwest (WN), and United Airlines (UA) in the top 10\% is at most 5\% higher than their proportion in the dataset.

%showcase an example for diversity over a dataset with more than a million records.
%Our objective is to find a ranking functions over three attributes {\tt departure delay}, {\tt arrival delay}, and {\tt taxi in} that diversify the results such that the percentage of outcomes from each of four major companies Delta Airlines (DL), American Airlines (AA), Southwest (WN), and United Airlines (UA) in the top 10\% of results is not more than 5\% higher than their percentage in the whole dataset.
We sample 1,000 records uniformly at random from the dataset of 1.3M records and use it for preprocessing with $N=40,000$.  Preprocessing took 1,276 seconds to complete.
Next, we used the complete dataset and checked if the function assigned to the cells using the sample  are in fact satisfactory. It turned out that for all assigned functions the percentage of 
results from each of four major airlines in the top 10\% was at most 5\% higher than their proportion in the whole dataset --- all of them were satisfactory.

\section{Related Work}\label{sec:related}

%\stitle{Fairness and diversity models:}
Several recent papers focus on measuring fairness in ranked lists\cite{DBLP:conf/ssdbm/YangS17,DBLP:conf/cikm/ZehlikeB0HMB17}, on constructing ranked lists that meet fairness criteria~\cite{DBLP:journals/corr/CelisSV17}, and on fair and diverse set selection~\cite{StoyanovichYJ18}.  Fairness in top-$k$ over a single binary type attribute (such as gender, ethnic majority/minority, or disability status) is studied in Zehlike et al.~\cite{DBLP:conf/cikm/ZehlikeB0HMB17}, where the goal is to ensure that the proportion of members of a protected group in every prefix of the ranking remains statistically above a given minimum.   
Celis et al.~\cite{DBLP:journals/corr/CelisSV17} provide a theoretical investigation of ranking with fairness constraints.  In their work, fairness in a ranked list is quantified as an upper bound on the number of items at the top-$k$ that belong to multiple, possibly overlapping, types.  In contrast, our goal is to assist the user in designing fair score-based rankers.  Our framework accommodates a large class of fairness constraints.  In our experiments, we focus on variants of fairness constraints similar to those in ~\cite{DBLP:journals/corr/CelisSV17,StoyanovichYJ18,DBLP:conf/cikm/ZehlikeB0HMB17}.

Diversification of query results has always been an important data retrieval topic~\cite{diversity,boyce1982beyond,agrawal2009diversifying}.
Different definitions of diversity include similarity function-based~\cite{chen2006less} and topic-based~\cite{agrawal2009diversifying}. General background on diversity and a connection to fairness are provided in~\cite{diversity}.
A nice property of the techniques proposed in this paper is that they are independent of the choice of a fairness function. In fact, one can replace the fairness oracle with any binary-output function that takes an ordering of the items as the input.
This makes our techniques suitable for a general range of diversity definitions.
%All is needed is to pass the orderings to the diversity oracle (instead of a fairness oracle), and all the claims in the paper remain valid.

%\stitle{Arrangement of hyperplanes and its applications:}
The techniques provided in this paper mainly follow the concepts in combinatorial geometry. The general background and the terms
%The background in combinatorial geometry and terms 
are provided in~\cite{de2008computational, edelsbrunner}.
In addition,~\cite{edelsbrunner} discussed the complexity bounds and proposes the incremental algorithm for constructing the lattice of arrangement.
Arrangement of hyperplanes is also studies in~\cite{orlik2013arrangements,grunbaum2003arrangements, schechtman1991arrangements}.
Applications of arrangements such as motion planner in robotics are discussed by P. Agrawal et. al.~\cite{agarwal2000arrangements}.

\section{Final Remarks}\label{sec:conclusion}

In this paper, we studied the problem of designing fair ranking schemes.
Considering the linear combinations of attribute values as the score of each item, our system assists users in choosing criterion weights that are fair. Creating proper indexes in an offline manner enables efficient answering of the users' queries. In addition to the theoretical analyses, empirical experiments on real datasets confirmed both efficiency and effectiveness of our proposal.

In this paper, we designed techniques for a general fairness definition that takes an ordering of the items as input and decides whether it meets the fairness requirements. Additional information about the fairness model can help optimize the techniques. For example, knowing that the fairness oracle investigates fairness only within the top-$k$ of the ordering~\cite{DBLP:conf/cikm/ZehlikeB0HMB17} can help in ignoring the items that do not belong to $k$ convex layers~\cite{chang2000onion}, as those will never appear within the top-$k$.
This reduces complexity of the arrangement from $n^{2(d-1)}$ to $n_k^{2(d-1)}$, where $n_k$ is the number of items in the top $k$ convex layers.  We will explore this and other kinds of optimizations in future work.
The techniques of this paper are provided for a fixed number of dimensions.
%For the variable high dimensional situations, the problem is NP-complete, following the complexity of the arrangement.
We consider extending our techniques to a variable number of dimensions for future work.
\submit{\newpage }
\bibliographystyle{abbrv}
\bibliography{ref}
\techrep{
\appendix
\section{Appendix}
\subsection{Angle distance computation}\label{appendix:angle}
As explained in \S~\ref{sec:pre}, linear ranking function can be represented as rays in $\mathbb{R}^d$ that start at the origin. These rays can be represented by $d-1$ angles.
Consider a ray $\rho$ that starts from the origin and passes through the point $p$.
Let $polar(p) = \langle r, \Theta\rangle$ be the polar representation of $p$.
First all the points $p'$ that $\rho$ passes through them have the polar representative $\langle r', \Theta\rangle$.
Second, for a point $p$ with the polar representative $\langle r, \Theta\rangle$, there is one and only one ray starting from the origin that passes through it.
%Second, given the angle combination $\Theta$, there is one and only one ray starting from the origin and passing through the point $p$ with the polar representation $(1,\Theta)$.
Thus, the angle vector $\Theta$ of size $d-1$ is enough for identifying this ray.
We use cosine similarity to compute the angle distance between two rays, represented by the angle vectors $\Theta^{(i)}$ and $\Theta^{(j)}$.
%The distance between the two angle vectors will not necessarily give the angle between them. For example consider the rays along with $y$ and $z$ axes in $\mathbb{R}^3$. The representation of these rays are $y: (\pi/2,0)$ and $z: (0,\pi/2)$, respectively. Looking at Figure~\ref{fig:hd1}, it is easy to see the angle between these rays is $\pi/2$.
%However, the $L_2$ norm between the vectors $(\pi/2,0)$ and $(0,\pi/2)$ is $\frac{\pi \sqrt{2}}{2}$.
%As another example, consider the angle between the rays $\rho_1(0,\pi/6)$ and $\rho_2(\pi/2,\pi/6)$, versus the rays $\rho_3(0,\pi/3)$ and $\rho_4(\pi/2,\pi/3)$ in Figure~\ref{fig:hd1}. While the $L_2$ norm between both $\rho_1$-$\rho_2$ and $\rho_3$-$\rho_4$ is $\pi/2$, both angles are smaller than $\pi/2$ and the angle between $\rho_1$-$\rho_2$ is larger then the one between $\rho_3$-$\rho_4$.

Consider the point $p_i = \langle 1, \Theta^{(i)} \rangle$ (that the ray $\Theta^{(i)}$ passes through it).
The cartesian coordinates of $p_i$ are\footnote{\small To simplify the equation, we set $\Theta_0^{(i)}$ to $\pi/2$, by appending it to the beginning of the vector $\Theta$.}:
\begin{align}
p_i = \langle \sin \Theta_k^{(i)} \underset{l=k+1}{\overset{d-1}{\Pi}} \cos \Theta_l^{(i)} ~,~ \forall 0\leq k < d \rangle
\end{align}
Using the definition of cosine similarity, for the points $p_i = \langle 1, \Theta^{(i)} \rangle$ and $p_j  = \langle 1, \Theta^{(j)} \rangle$: 
\begin{align}
\cos (\theta_{ij}) = \sum\limits_{k=0}^{d-1} \sin \Theta_k^{(i)} \sin \Theta_k^{(j)} \underset{l=k+1}{\overset{d-1}{\Pi}} ( \cos \Theta_l^{(i)} \cos \Theta_l^{(j)})
\end{align}
Thus, $\theta_{ij}$ (the angle between the rays $\Theta^{(i)}$ and $\Theta^{(j)}$) is:
\begin{align}\label{eq:raydistance}
\theta_{ij} = \arccos\big( \sum\limits_{k=0}^{d-1} \sin \Theta_k^{(i)} \sin \Theta_k^{(j)} \underset{l=k+1}{\overset{d-1}{\Pi}} ( \cos \Theta_l^{(i)} \cos \Theta_l^{(j)}) \big)
\end{align}

\subsection{Angle space partitioning}\label{subsec:anglepartitioning}
According to Appendix~\ref{appendix:angle}, the distance between two rays specified by two $(d-1)$ dimensional angle vectors $\Theta^{(i)}$ and $\Theta^{(j)}$ is not the same as their euclidean distance.
Thus, as also discussed in~\cite{vlachou2008angle}, a regular grid partitioning that equally partitions each axis into $\sqrt[d]{N}$ equal size ranges will not generate cells of equal sizes.
One can verify this by looking at Figure~\ref{fig:hd1}, in which the cells in the bottom row have larger areas than the ones in the upper rows.
Inspired by~\cite{vlachou2008angle}, we propose the angle space partitioning that partitions the space into $N$ equal area cells.
We do the partitioning using the surface of (the first quadrant of) the unit hypersphere in $\mathbb{R}^d$.
Consider a hypercone starting from the origin, while its base is a hypercube (square in $\mathbb{R}^3$) on the surface of unit hypersphere.
%The space of rays (functions) inside a cell in the partitioned grid form a cone that its base in the origin and it base is hyper-square.
%For example, Figure~\ref{sss} shows such a cone in $\mathbb{R}^3$.
We want to partition the space into $N$ such hypercones such that the area of all cells are equal.
The total area of the space (the area of the first quadrant of the unit hypersphere) is~\footnote{\scriptsize http://mathworld.wolfram.com/Hypersphere.html}
\begin{align}
\eta = \frac{\pi^{d/2}}{2^{d-1} \Gamma (d/2)}
\end{align}
where $\Gamma$ is the gamma function.
Thus, the area of each cell is
\begin{align}\label{eq:cellarea}
\eta_{\mbox{cell}} = \frac{\pi^{d/2}}{N 2^{d-1} \Gamma (d/2)}
\end{align}
Considering the cells to be small enough, one can assume that the area of each cone on the surface of the hypersphere is equal to the area if its base.
The area of a hypercube with sides of size $\gamma$ is $\gamma^{d-1}$ (e.g. $\gamma^2$ in $\mathbb{R}^3$).
Assuming the area of the hypercone and its base to be equal, using Equation~\ref{eq:cellarea}, the sides of the hyper square are of size:
\begin{align}\label{eq:cellside}
\sqrt[d-1]{\eta_{\mbox{cell}}} = \sqrt[d-1]{ \frac{\pi^{d/2}}{N 2^{d-1} \Gamma (d/2)} }
\end{align}
Since the radius of the hypersphere is $1$, the angle between the rays in two corners of a side are: 
\begin{align}\label{eq:cellside}
\gamma = 2\arcsin \frac{\sqrt[d-1]{ \frac{\pi^{d/2}}{N 2^{d-1} \Gamma (d/2)} } }{2}
\end{align}
We use $\gamma$, as computed in Equation~\ref{eq:cellside}, for angle space partitioning, as follows.

Consider the axes $\theta_1,\theta_2,\cdots , \theta_{d-1}$.
For every axis, we maintain a vector of angles $T_{\theta_i}$ such that each element $T_{\theta_i}[j]$ of the vector contains:
\begin{itemize}
\item {\em range:} the borders of the row in axis $\theta_i$.
\item {\em elements:} the vector of angles for axis $T_{\theta_{i+1}}$ in row $T_{\theta_i}[j]$.
\end{itemize}

One can see the partitioning data structure as a tree of depth $d-1$ that its leaves are the cells; the path from the root to each leaf identifies its borders in every dimension.
%\textcolor{blue}{Figure...}
%\missingfigure{Show an example tree for partitioning the angle space}
In order to construct the ranges, in an iterative manner, we apply Equation~\ref{eq:raydistance} to specify ranges of angle $\gamma$ as the rows of each axis. Then, we recursively partition the rows of the axis into equal area cells.
Algorithm~\ref{alg:partition} shows the pseudo code for angle space partitioning.
Consider the moment where the algorithm is partitioning the $i$-th axis and the current point is in the form of $p_c = \langle \Theta_1, \Theta_2, \cdots,\Theta_{i-1}, \theta, 0,\cdots ,0\rangle$.
The objective is to find the next point in $i$-th axis such that the angle of its corresponding ray with the current point is $\gamma$.
The next point is in the form of $p_n = \langle \Theta_1, \Theta_2, \cdots,\Theta_{i-1}, \theta',0,\cdots ,0\rangle$, where $\theta'$ is unknown.
Using Equation~\ref{eq:raydistance}, the angle between the rays of $p_c$ and $p_n$ can be rewritten as:
\begin{align}
\cos \gamma = \cos \theta' \cos \theta \sum\limits_{k=0}^{i-1} \sin^2 \Theta_k \underset{l=k+1}{\overset{i-1}{\Pi}} ( \cos^2 \Theta_l) + \sin \theta' \sin \theta
\end{align}
Let $\alpha$ be $\cos \theta \sum\limits_{k=0}^{i-1} \sin^2 \Theta_k \underset{l=k+1}{\overset{i-1}{\Pi}} ( \cos^2 \Theta_l)$ and $\beta$ be $\sin \theta$.
Then the angle between the above equation is
\[
	\alpha \cos \theta' + \beta \sin \theta' = \cos \gamma
\]
Now, let us set $\delta = \arctan\frac{\beta}{\alpha}$ and $\Delta = \sqrt{\alpha^2 + \beta^2}$.
Thus,
\begin{align}\label{eq:getnextangle}
\nonumber
& \Delta \cos \delta \cos \theta' + \Delta \sin \delta \sin \theta' = \cos \gamma \\
\nonumber
& \Rightarrow \Delta \cos(\theta' - \delta) = \cos \gamma \\
& \Rightarrow \theta' = \arccos \frac{\cos \gamma}{\Delta} + \delta
\end{align}

\begin{algorithm}[!h]
\caption{{\bf \anglepartitioning}\\
		 {\bf Input:} axis number $i$, angle combination for previous axes $\Theta$, $d$\\
		 {\bf Output:} Partitioned space $T$
		}
\begin{algorithmic}[1]
\label{alg:partition}
    \STATE $\theta = 0$, $T=\{\}$, $j=1$
    \WHILE{$\theta<\pi/2$}
    	\STATE compute $\theta'$, using Equation~\ref{eq:getnextangle}
    	\STATE $T[j].$range$=(\theta,\theta')$
    	\IF{$i<(d-1)$}
    		\STATE $\Theta[i] = \theta$
    		\STATE $T[j].$elements$=${\bf \anglepartitioning}$(\Theta,i+1,d)$
    	\ENDIF
    	\STATE $\theta = \theta'$, $j=j+1$
    \ENDWHILE
    \STATE {\bf return} leaves($T$)
\end{algorithmic}
\end{algorithm}

\begin{theorem}
Algorithm~\ref{alg:partition} is in $O(N)$.
\end{theorem}
\begin{proof}
The total number of cells is $N$.
Looking at the recursion tree, every leaf of the tree (every cell) has the level $d$. Therefore, for a constant value of $d$, the cost of generating each cell is constant. Therefore, Algorithm~\ref{alg:partition} is in $O(N)$.
\end{proof}

}
\end{document}